\documentclass[authoryear,preprint,review]{elsarticle}

\usepackage{amssymb}
 \usepackage{amsthm}
\usepackage{amsmath}

\usepackage{lscape}
\usepackage{multirow}
\usepackage{amsfonts,booktabs,color,epsfig,graphicx,url}
\usepackage{comment}
\usepackage[pdftex,colorlinks=true,hypertexnames=false]{hyperref}
\definecolor{darkblue}{rgb}{0,0,.6}
\hypersetup{citecolor=darkblue,linkcolor=darkblue,urlcolor=darkblue}
\hypersetup{pdfauthor={LLZZ}}

\DeclareMathOperator*{\argmin}{\arg\min}

\theoremstyle{plain}
\newtheorem{theorem}{Theorem}

\newtheorem{proposition}{Proposition}
\newtheorem{lemma}{Lemma}

\theoremstyle{definition}
\newtheorem{definition}{Definition}
\newtheorem{remark}{Remark}

\newtheorem{assumption}{Assumption}

\journal{}

\usepackage{geometry}
\begin{document}

\begin{frontmatter}


\title{Inference for high-dimensional linear expectile regression with de-biased method}

\author[a]{Xiang Li}
\ead{12135003@zju.edu.cn}
\affiliation[a]{organization={Zhejiang University},
            addressline={866 Yuhangtang Rd}, 
            city={Hangzhou},
            postcode={310058}, 
            country={China}}
 \author[b]{Yu-Ning Li}
\ead{yuning.li@york.ac.uk}
\affiliation[b]{organization={School for Business and Society, University of York},
            addressline={Heslington}, 
            city={York},
            postcode={YO10 5DD}, 
            country={United Kingdom}}
 \author[a]{Li-Xin Zhang}
 \ead{stazlx@zju.edu}
 \author[c]{Jun Zhao\corref{cor1}}
\ead{zhaojun@hzcu.edu.cn}
 \cortext[cor1]{Corresponding author.}
\affiliation[c]{organization={School of Computer and Computing Science, Hangzhou City University},
            addressline={No. 48 Huzhou Street}, 
            city={Hangzhou},
            postcode={310015}, 
            country={China}}

\title{}

\begin{abstract}
In this paper, we address the inference problem in high-dimensional linear expectile regression. We transform the expectile loss into a weighted-least-squares form and apply a de-biased strategy to establish Wald-type tests for multiple constraints within a regularized framework. Simultaneously, we construct an estimator for the pseudo-inverse of the generalized Hessian matrix in high dimension with general amenable regularizers including Lasso and SCAD, and demonstrate its consistency through a new proof technique. We conduct simulation studies and real data applications to demonstrate the efficacy of our proposed test statistic in both homoscedastic and heteroscedastic scenarios.
\end{abstract}




\end{frontmatter}


\section{Introduction}\label{sec:1}

High-dimensional datasets have become increasingly prevalent in many fields, such as finance and genetics, presenting significant challenges for traditional regression methods such as least squares. In particular, these methods can lead to over-fitting and unreliable results in the high-dimensional setting, that is when the number of variables exceeds the sample size. By  imposing sparsity constraints on the model coefficients, regularization techniques, such as the Lasso, have been developed to address these limitations. Nonetheless, sparsity constraints bring in complexities for high-dimensional inference, calling for further investigation and tailored methodologies.

 Despite notable advancements made in the area of high-dimensional estimation and inference, a considerable portion of the existing literature primarily focuses on regression in mean, providing insights solely into the conditional mean of the response variable. 
To offer a more comprehensive understanding of the underlying distributional information, alternative regression methods such as quantile regression and expectile regression \citep[e.g.,][]{Koe05,NP87} have been developed to estimate the quantiles and expectiles, respectively, of the response variable. In recent years, there has been a growing interest in expectile regression applied to high-dimensional data, leading to the development of several regularized expectile regression methods. For instance, \cite{WOP19} introduce an $L_0$-regularized expectile regression model with a Whittaker smoother, specifically tailored for modelling accelerometer data. \cite{GZ16} study the expectile regression with Lasso ($L_1$ regularizer) and  nonconvex penalties. Further contributions to this domain include the works of 
 \cite{ZCZ18}, \cite{ZZ18}, \cite{LPC19}, \cite{Ciu21} and \cite{XDJYS21}, which study the expectile regression with SCAD, adaptive Lasso, and elastic-net regularizers and show the oracle properties of the estimators. Recently, \cite{MTWZ23} study the  robust expectile regression in high dimensions using the generalized Huber loss with the reweighted $L_1$ regularizer.

While the estimation of high-dimensional expectile regression has been well-studied, the statistical inference of such models remains an active area of investigation.
The existing literature on the inference of expectile regression primarily focuses on the low-dimensional setting since \cite{NP87}. 
\cite{ZZ18} develop de-biased estimation for inference of expectile regression when the number of variables is finite. Additionally,
\cite{JPD21} extend it to the single-index expectile model. 
Recently, \cite{SLZ21} consider the inference for large-scale data using the divide and conquer algorithm.  Furthermore, \cite{BDT23} and \cite{ZXY22} develop a sup-Wald test for Granger non-causality in an expectile range, which allows for the examination of Granger causality in the entire distribution by jointly testing for Granger non-causality at all expectiles. 

When it comes to the high-dimensional settings, the penalized estimators suffer from non-negligible bias due to the effect of the massive dimensionality,  which makes them unfeasible to be used  directly for inference. To overcome such a problem, \cite{VBRD14}, \cite{ZZ14}  and \cite{JM14} introduce the so-called 'de-biased' or 'de-sparsified' procedure with penalized least squares to correct the bias of the initial estimator so that the de-biased estimator can be used for inference purpose. Inspired by such an idea, many recent papers have focused on the relevant generalizations for the  de-biased approach.  \cite{CG17} study the optimal expected lengths of confidence intervals for general linear functions of a high-dimensional regression vector from both the minimax and the adaptive perspective.  \cite{DBZ17} propose some bootstrap methodology for individual and simultaneous inference in high-dimensional linear models with possibly non-Gaussian and heteroscedastic errors.  \cite{CGM21} study the GLMs with binary outcomes via the optimization in  quadratic form.   Along this line of research,  \cite{CCK22} further make an extension to the time series framework while \cite{LZT23} combine the de-biased  approach with transfer learning.
 However, to our knowledge, no literature has discussed the statistical inference for the high dimensional GLMs in the context of the expectile regression.

 In this paper, we aim to extend the study of the inference of linear expectile regression to the high-dimensional setting,  which allows the number of covariates to grow with the sample size and diverge to infinity. The main contributions are four-folds. First and foremost, unlike the equal-weight methods for bias correction in the high-dimensional mean regression models \citep[e.g.,][]{ZZ14, JM14}, we develop the bias correction procedure for high-dimensional expectile regression models with the expectile-specified random weights. 
Such a weighting strategy is essential to the de-biased procedure of the inference problems in high-dimensional generalized linear models, e.g., \cite{VBRD14,CGM21}, and can be of independent interest to study. 
 Furthermore, 
we extend the study by \cite{VBRD14} who consider a twice differentiable and local Lipstchiz loss function which excludes the expectile loss function. We develop an alternative proving strategy to show that the estimation errors in the preliminary estimator have little impact on the estimation of the expectile-specified random weights.
Secondly, we extend the choice of the regularizer to a broader category, that is the amenable regularizers discussed by \cite{LW15} and \cite{Loh17}.  This extension provides an avenue for the application of a wider range of penalties, such as the SCAD regularizer, which exhibits favorable properties for high-dimensional model selection. 
  Thirdly, our paper also delves into precision matrix estimation in a high-dimensional setting, further advancing previous research by incorporating random weights and considering estimation errors in those weights. This extends the results by \cite{FFW09} and \cite{LW15} which also utilize the non-convex regularizer in the precision matrix estimation. 
  Finally, we establish a test statistic for multivariate testing within the high-dimensional expectile regression framework. This test is applied in two the empirical study cases, one in finance, revealing that momentary supply may possess predictive power in stock returns, the other in Genomics, screening important genes related to the target gene `TLR8'.
 
The rest of the article is organized as follows. In Section~\ref{sec2}, we introduce the high-dimensional expectile regression and the inference problem. In Section~\ref{sec3}, we introduce the expectile-specified weighting matrix and propose the de-biased estimator of penalized expectile regression. Then we approximate the   pseudo-inverse of generalized Hessian matrix by the node-wise regression. In Section~\ref{sec4}, we establish a Wald-type test based on the asymptotic normality of the de-biased estimator. Section~\ref{sec5} provides numerical results.
Section~\ref{sec6} applies the proposed method to a financial dataset and a genetic dataset.
Section~\ref{sec7}  contains more discussion along with possible future works. 

{\bf Notations}. 
For a $p- $dimensional vector $ \boldsymbol{v} = (v_1,...,v_p)^\top \in \mathbb{R}^p $, and $1\leq q < \infty $, we define $\Vert \boldsymbol{v} \Vert_{q} = (\sum_{i=1}^p |v_i|^q)^{1/q} $, and particularly, we denote by 
 $\Vert \boldsymbol{v} \Vert_{\infty} = \max_{1\leq i\leq p}|v_i| $ the infinity norm of a vector. Denote $\Vert \boldsymbol{v} \Vert_{0} = | \text{supp}(\boldsymbol{v} ) | $, where $\text{supp}(\boldsymbol{v} ) = \{i; v_i \neq 0 \} $ is the active set. 
Moreover, we denote by $\langle \boldsymbol{u},\boldsymbol{v}\rangle = \boldsymbol{u}^\top\boldsymbol{v} $ the inner product of $\boldsymbol{u},\boldsymbol{v} \in \mathbb{R}^p $.
Furthermore, for ${\cal A} \subseteq \left\{1,...,p \right\} $, denote $|{\cal A} | $ by the cardinality of $\cal A$. We let $\boldsymbol{v}_{\cal A} = (v_i)_{i\in {\cal A}}$  and ${\cal A}^C $ be the complement of $\cal A$.   For a differentiable function $g: \mathbb{R}^p  \to \mathbb{R} $, we denote by $\nabla g(\boldsymbol{v}) = \partial g(\boldsymbol{v})/\partial \boldsymbol{v} $ and $ \nabla_{\cal A} g(\boldsymbol{v}) = \partial g(\boldsymbol{v})/\partial \boldsymbol{v}_{\cal A} $.    For a matrix $\mathbf{Q} = (Q_{ij}) $, we denote by $\mathbf{Q}^\top $ the transpose of the matrix $\mathbf{Q} $, $ \Vert \mathbf{Q} \Vert_{\infty} = \max_{ij} |Q_{ij}|$ the element-wise sup-norm, $\Vert \mathbf{Q} \Vert_{l_1} = \max_{j}\sum_{i} |Q_{ij}| $ the the $l_1$ norm, and $\Vert \mathbf{Q} \Vert_{l_\infty} = \max_{i}\sum_{j} |Q_{ij}| $ the the $l_\infty$ norm. For symmetric matrix $\mathbf{Q} $, we denote by $\lambda_{\min}(\mathbf{Q} ) $ and $\lambda_{\max}(\mathbf{Q} ) $ its minimal and maximal eigenvalues. Particularly, we denote by $\mathbf{I} $ the unit matrix and by $\boldsymbol{e}_j$ its $j-$th column. Let $a\vee b $ and $a\wedge b $ denote $\max\{a,b\} $ and $\min\{a,b\} $, respectively; and let $a_n\asymp b_n $ denote that $a_n=O(b_n) $ and $b_n= O(a_n) $ hold jointly.


\section{High-dimensional expectile regression and the inference problem}\label{sec2}
We start by introducing the asymmetric square loss function, or the so-called expectile loss function (e.g.,  \cite{NP87} and \cite{Efr91}), 
\begin{equation*}
   \rho_{\tau}(u) = |\tau - \mathbb{I}(u<0)|u^2 = \left\{\begin{aligned} \nonumber
    &\tau u^2, ~~~~\qquad u \ge 0,\\
    &(1-\tau) u^2,\quad u < 0,\\
   \end{aligned}
    \right.
\end{equation*}   
where $\tau \in \left(0,1\right) $ is a positive constant. 
Then the $\tau- $th expectile of a random variable $Y $ can be defined as 
\begin{equation*}
 m_{\tau} (Y) = \argmin_{m} {\sf E}[ \rho_{\tau}(Y - m)].
\end{equation*} 
The expectile loss function assigns different weights on the squared loss associated with the sign of the residual $Y- m_{\tau} (Y) $.
It is easy to check that $m_{\tau}(Y - m_{\tau} (Y))=0 $ and $m_{1/2}(Y)={\sf E}[Y] $. It is also worth mentioning that the expectile loss function is not second-order differentiable due to the discontinuity of the first-order derivative at 0.



Next, we introduce the high-dimensional expectile framework. Consider a set of $n$ independent and identically distributed multivariate random variables $\{y_{i},\boldsymbol{X}_{i} \}, i \in \{1,...,n\}$, where $ \boldsymbol{X}_{i}  = (x_{i1},...,x_{ip})^\top $ is a $p$-dimensional covariates and $y_{i} $ is a scalar variable, where the number of covariates $p$ can be larger then the sample size $n$. We assume that they are generated by the following high-dimensional expectile linear model:
\begin{equation}\label{hlm01}
  y_{i}  = \boldsymbol{X}_{i}^\top \boldsymbol{\beta}^*   + \epsilon_{i},
\end{equation}
  where $\boldsymbol{\beta}^*  = \argmin_{\boldsymbol{\beta}   \in \mathbb{R}^p} {\sf E} [\rho_{\tau}(Y_{i} - \boldsymbol{X}_{i}^\top\boldsymbol{\beta}  )]$ is a $p $-dimensional vector of parameters, $\epsilon_{i} $ is the error term.  We omit the subscript $
\tau$ of $\boldsymbol{\beta}^*$ hereafter when
there is no confusion.   For identification purpose, we assume that $\{\epsilon_{i}\},i = 1,...,n\ $ are i.i.d distributed and satisfy $m_{\tau}(\epsilon_{i}|\boldsymbol{X}_{i})=0 $, or equivalently ${\sf E} \left[\rho'_{\tau}(\epsilon_{i}) |\boldsymbol{X}_{i}\right] = 0 $.  
This framework allows a general conditional heteroscedastic setting or the so-called location-scale framework, that is $\epsilon_{i}=\sigma(\boldsymbol{X}_{i})z_i$, where $z_i$ is the innovation that is independent of $\boldsymbol{X}_{i}$ and the scale function $\sigma(\boldsymbol{X}_{i})$ can be either of a linear form (e.g.,  \cite{AM87} and \cite{GZ16})  or a non-parametric form (e.g., \cite{RRS96} and \cite{FY98}). 

  Under this high-dimensional framework, a popular approach to estimate the coefficient $\boldsymbol{\beta} $ is given by the regularized  asymmetric-least-squares (ALS),
\begin{equation} \label{est01}
  \hat{\boldsymbol{\beta}  } = \argmin \limits_{\Vert\boldsymbol{\beta}\Vert_1\leq R } L_n(\boldsymbol{\beta}  ) + P_{\lambda}(\boldsymbol{\beta}  ) = \argmin \limits_{\Vert\boldsymbol{\beta}\Vert_1\leq R} \dfrac{1}{2n}\sum_{i=1}^n \rho_{\tau}(y_i - \boldsymbol{X}_i^\top  \boldsymbol{\beta}  ) +P_{\lambda}(\boldsymbol{\beta}  ) ,
\end{equation}
where $L_n(\boldsymbol{\beta}  )$ is the expectile loss function and the constant $1/2$ is introduced for the sake of simplicity in the statistical results. 
Regarding the regularizer $P_{\lambda}(\boldsymbol{\beta}  ) $ with a tuning parameter $\lambda>0$, we consider the amenable regularizer which is developed by \cite{ LW15} and formally stated in \cite{LW17}. This category encompasses both convex and non-convex regularizers, including the Lasso, MCP, SCAD, among others. Due to the potential non-convexity of the regularizer, following \cite{LW15, LW17}, we add the constraint $\Vert \boldsymbol{\beta} \Vert_1 \leq R$ in order to ensure that a global minimum $\hat{\boldsymbol{\beta}}$
exists. The boundary $R$ is taken as another tuning parameter which is chosen to satisfy $ \Vert \boldsymbol{\beta}^* \Vert_1 \leq R $ to ensure $\boldsymbol{\beta}^*$ is feasible. Specifically,  the amenable regularizer is defined as follows:

\begin{definition}\label{amenable-p}
  (Amenable regularizer). Suppose that the regularizer $P_{\lambda}(\boldsymbol{\beta})$ is separable across each coordinate, that is,
  \begin{equation*}
    P_{\lambda}(\boldsymbol{\beta}) = \sum _{j = 1}^{p} p_{\lambda}(\beta_{j}), 
  \end{equation*}
  where $p_{\lambda}: \mathbb{R} \to \mathbb{R} $ is some scalar function. If $ p_{\lambda}$  satisfies the following conditions (i)--(vi), we say that $P_{\lambda}$ (or $p_{\lambda}$) is $\mu $-amenable. Moreover, if $ p_{\lambda}$ further satisfies the condition (vii), we say that $P_{\lambda}$ (or $p_{\lambda}$) is $(\mu,\gamma)-$amenable.
  \begin{enumerate}[(i)]
    \item The scalar function $p_{\lambda}(\cdot) $  is symmetric around zero and $p_{\lambda} (0)= 0 $.
    \item The scalar function $p_{\lambda}(\cdot) $  is non-decreasing on $\mathbb{R}^{+} $.
    \item The scalar function $\dfrac{p_{\lambda}(t)}{t} $  is non-increasing on $\mathbb{R}^{+} $.
    \item The scalar function $p_{\lambda}(t) $  is differentiable for all $t\neq 0 $.
    \item $\lim_{t \to 0^{+}} p'_{\lambda}(t) = \lambda. $
    \item There exists a constant $\mu > 0 $ such that the scalar function $p_{\lambda}(t) + \dfrac{\mu}{2}t^{2} $ is convex.
    \item There exists a constant $\gamma >0 $ such that $p'_{\lambda}(t)=0 $ for all $t \geq \gamma \lambda $.
  \end{enumerate} 
\end{definition}


 The conditions (i)--(iii) are relatively mild and feasible for a large variety of regularizers, which is also discussed in \cite{ZZ12}. 
The conditions (iv) and (v) exclude some unfeasible regularizers,  such as the capped $l_1$ regularizer, which is not differentiable on many points on the positive real line, and the ridge 
regularizer due to its behavior at the origin. 
The condition (vi) represents for the weak convexity which is proposed by \cite{VJ82}, and it can be viewed as a curvature constraint that controls the non-convexity of the regularizer.
Moreover, the condition (vii) ensures the oracle property of the corresponding estimator and a similar condition can be found in \cite{WKL14}. 
More discussion and properties of the amenable regularizers can be found in \cite{Loh17} and \cite{LW15, LW17}.

\begin{remark}
The Lasso regularizer, $p_{\lambda} (u)= |u| $, is convex and $0- $amenable. However, it is  not $(0,\gamma)- $amenable for any $0<\gamma<\infty$. 
\end{remark}
\begin{remark}The Smoothly clipped absolute deviation (SCAD) regularizer is non-convex, following \cite{FL01}, takes the form 
 $$
  p_{\lambda}(u) = \lambda|u|\cdot \mathbb{I}(|u|\leq \lambda)+ \dfrac{2a\lambda|u| - u^2-\lambda^2}{2(a-1)} \cdot\mathbb{I}( \lambda < |u|\leq a\lambda)+ \dfrac{(a+1)\lambda^2}{2}\cdot \mathbb{I}( a\lambda < |u|),
 $$
  where $a>2 $ is a fixed constant. Moreover, the SCAD regularizer is $(\mu,\gamma)- $amenable, with $\mu = \dfrac{1}{a-1} $ and $\gamma = a $.
\end{remark}

Our main interest falls on the statistical inference for the parameter vector $\boldsymbol{\beta}  ^* $ in the high-dimensional expectile linear regression model (\ref{hlm01}). To be specific, given a certain expectile level $\tau $, we want to test, for a given $p_0
$-by-$p$ full row-rank hypothesis matrix $\mathbf{R}$ and a $p_0$-dimensional vector $\boldsymbol{c}$,
\begin{equation}\label{test01}
  {\sf H}_{0}: \mathbf{R}{ \boldsymbol\beta}^* = \boldsymbol{c} \quad \text{versus} \quad {\sf H}_{1}: \mathbf{R}{ \boldsymbol\beta}^*  \neq \boldsymbol{c},
\end{equation}
where $p_0$ is the number of constraints. For example, when $\mathbf{R}=(1,0,0,\ldots,0)$ and $\boldsymbol{c}=0$, we have ${\sf H}_{0}: {\beta}_1^*=0$. When $\mathbf{R}=(1/p,1/p,\ldots,1/p)$ and $\boldsymbol{c}=0$, we have ${\sf H}_{0}: \sum_{j=1}^p{\beta}_j^*/p=0$.

\section{The de-biased method under the expectile framework.}\label{sec3}
In this section, we establish a test statistic for hypotheses (\ref{test01}) by developing the de-biased method under the expectile framework.


\subsection{The de-biased estimator.}\label{sec3.1}

To begin with, we introduce a $\boldsymbol{\beta}  - $related weighting matrix $\mathbf{W}_{\boldsymbol{\beta}  } = {\sf diag}(w_{\boldsymbol{\beta} ,1},...,w_{\boldsymbol{\beta} ,n}) $, where 
\begin{equation*}
  w_{\boldsymbol{\beta}  ,i} = |\tau - \mathbb{I}(y_{i}-\boldsymbol{X}_{i}^\top \boldsymbol{\beta}   <0 )|^{1/2}.
\end{equation*}
Consequently, the expectile loss function in (\ref{est01}) can be rewritten as follows,  
\begin{equation*}
  L_n(\boldsymbol{\beta}  ) = \dfrac{1}{2n}\sum_{i=1}^n w^2_{\boldsymbol{\beta}  ,i}(y_{i} -\boldsymbol{X}_{i}^\top \boldsymbol{\beta}   )^2  = \dfrac{1}{2n}\Vert (\mathbf{W}_{\boldsymbol{\beta}  }\boldsymbol{Y} - \mathbf{W}_{\boldsymbol{\beta}  }\mathbf{X}\boldsymbol{\beta}  )\Vert^2_{2},
\end{equation*}
where $\boldsymbol{Y}=(y_1,\ldots,y_n)^\top$ and $\mathbf{X}=(\boldsymbol{X}_{1},\ldots,\boldsymbol{X}_{n})^\top$.

Featured by the weighting matrix $\mathbf{W}_{\boldsymbol{\beta}  }$, the expectile loss function shares a similarity with the Generalized Least Squares (GLS) loss function.  However, it is noteworthy to point out that the weighting matrix depends on $\boldsymbol{\epsilon} $ and also depends on $\mathbf{X}$ when $\boldsymbol{\beta}  \neq\boldsymbol{\beta}  ^* $. 
Then inspired by the de-biased LASSO proposed by \cite{VBRD14}, we introduce the de-biased estimator of regularized expectile regression (e.g., \cite{GZ16} and \cite{ZCZ18}) as follows,
\begin{equation}\label{debias-est}
  \hat{\boldsymbol{\beta}  }_{de} = \hat{\boldsymbol{\beta}  }+ \hat{\boldsymbol{\Theta}}_{\hat{\boldsymbol{\beta}  }} (\mathbf{W}_{\hat{\boldsymbol{\beta}  }}\mathbf{X})^\top\mathbf{W}_{\hat{\boldsymbol{\beta}  }}\hat{\boldsymbol{\epsilon}}/n,
\end{equation} 
 where $\hat{\boldsymbol{\epsilon}} = \boldsymbol{Y}-\mathbf{X}\hat{\boldsymbol{\beta}  }$ and $\hat{\boldsymbol{\Theta}}_{\hat{\boldsymbol{\beta}  }} $ is an estimate of the inverse of the covariance matrix (i.e., the precision matrix) of $w_{\boldsymbol{\beta}^*  ,i}\boldsymbol{X}_i$. Obtaining $\hat{\boldsymbol{\Theta}}_{\hat{\boldsymbol{\beta}  }} $ poses a twofold challenge. Firstly, $w_{\boldsymbol{\beta}^*  ,i}$ is not directly observable, and we must estimate it using $w_{\hat{\boldsymbol{\beta}},i}$. Secondly, given that $p>n$, the (weighted) sample covariance matrix is non-invertible. While the second challenge is well-documented and solved in the literature, the first one requires further attention, which we address in Section \ref{sec3.2}.
To be more specific, the de-biased estimator allows the following decomposition,
 \begin{equation}\label{de-bias01}
  \begin{aligned}
  \hat{\boldsymbol{\beta}  }_{de} - \boldsymbol{\beta}  ^*  &=  \hat{\boldsymbol{\Theta}}_{\hat{\boldsymbol{\beta}  }} \mathbf{X}^\top\mathbf{W}_{{\boldsymbol{\beta}  }^*}^2\boldsymbol{\epsilon}/n - (\hat{\boldsymbol{\Theta}}_{\hat{\boldsymbol{\beta}  }} \mathbf{X}^\top\mathbf{W}_{{\boldsymbol{\beta}  }^*}^2\boldsymbol{\epsilon}/n- \hat{\boldsymbol{\Theta}}_{\hat{\boldsymbol{\beta}  }} \mathbf{X}^\top\mathbf{W}_{\hat{\boldsymbol{\beta}  }}^2\boldsymbol{\epsilon}/n) - (\hat{\boldsymbol{\Theta}}_{\hat{\boldsymbol{\beta}  }}\hat{{\boldsymbol{\Sigma}}}_{\hat{\boldsymbol{\beta}  }}-\mathbf{I})(\hat{\boldsymbol{\beta}  } - \boldsymbol{\beta}  ^*)\\
  & :=\hat{\boldsymbol{\Theta}}_{\hat{\boldsymbol{\beta}  }} \mathbf{X}^\top\mathbf{W}_{{\boldsymbol{\beta}  }^*}^2\boldsymbol{\epsilon}/n- {\boldsymbol{ \Delta}}^{(1)} -  {\boldsymbol{ \Delta}}^{(2)},
  \end{aligned}
    \end{equation}
  where 
  \begin{equation*}
 \hat{{\boldsymbol{\Sigma}}}_{\boldsymbol{\beta}  } = \mathbf{X}^\top_{\boldsymbol{\beta}  }\mathbf{X}_{\boldsymbol{\beta}  } /n, \quad \text{and} \quad  \mathbf{X}_{\boldsymbol{\beta}  } = \mathbf{W}_{\boldsymbol{\beta}  } \mathbf{X}. 
\end{equation*}
  The term ${\boldsymbol{ \Delta}}^{(1)} $  is related to the approximation error of the weighting matrix $\mathbf{W}_{\hat{\boldsymbol{\beta}  }}^2 $ with respect to $\mathbf{W}_{{\boldsymbol{\beta}  }^*}^2 $ and  the term ${\boldsymbol{ \Delta}}^{(2)} $ is related to the approximation of the inverse of the covariance matrix and the  bias of the estimation $(\hat{\boldsymbol{\beta}  } - \boldsymbol{\beta}  ^*) $.
  To utilize the proposed estimator for inference, it is essential to establish the asymptotic normality of $\hat{\boldsymbol{\Theta}}_{\hat{\boldsymbol{\beta}}} \mathbf{X}^\top\mathbf{W}_{{\boldsymbol{\beta}  }^*}^2\boldsymbol{\epsilon}/\sqrt{n} $ and demonstrate the negligibility of $\sqrt{n}\Vert{\boldsymbol{ \Delta}}^{(1)}\Vert_{\infty} $ and $\sqrt{n} \Vert{\boldsymbol{ \Delta}}^{(2)}\Vert_{\infty}$. 

\subsection {The pseudo-inverse from node-wise regression }\label{sec3.2}

In this section, we employ a widely-used node-wise regression method (see  \cite{VBRD14} and \cite{MNB06}) to construct $\hat{\boldsymbol{\Theta}}_{\hat{\boldsymbol{\beta}  }} $, which also serves as a pseudo-inverse of $\hat{{\boldsymbol{\Sigma}}}_{\hat{\boldsymbol{\beta}  }} $.  

Let $\boldsymbol{X}_{{\boldsymbol{\beta}  },(j)}$ denote the $j$-th column of the weighted design matrix $\mathbf{X}_{{\boldsymbol{\beta}  }} $, and $\mathbf{X}_{{\boldsymbol{\beta}  },(-j)} $ represent the weighted design matrix without the $j$-th column. To handle the high dimensionality, we apply node-wise regression within the classic regularized framework. Specifically, for each $j\in\{1,\ldots, p\}$, we seek the solution $\hat{\boldsymbol{\varphi}}_{\hat{\boldsymbol{\beta}},j} $ that minimizes the following objective function:
\begin{equation}\label{op-de}
\hat{\boldsymbol{\varphi}}_{\hat{\boldsymbol{\beta}  },j} := \argmin_{\Vert\boldsymbol{\varphi}\Vert_1\leq R_j}  \Vert \boldsymbol{X}_{\hat{\boldsymbol{\beta}  },(j)} - \mathbf{X}_{\hat{\boldsymbol{\beta}  },(-j)} \boldsymbol{\varphi} \Vert^2_{2}/(2n) + Q_{\lambda_{j}}(\boldsymbol{\varphi} ),   
\end{equation}
where $\hat{\boldsymbol{\varphi}}_{\hat{\boldsymbol{\beta}  },j} $ is a $(p-1) $-dimension vector with elements $ \hat{{\varphi}}_{\hat{\boldsymbol{\beta}  },jl}, l\in\{1,\ldots, p\}, l \neq j $, and $Q_{\lambda_{j}}(\cdot) $ is the regularizer with tuning parameter $\lambda_{j}$. Again,
following \cite{LW15}, the constraint $\Vert \boldsymbol{\varphi} \Vert_1 \leq R_j$  ensures a global minimum. 
By letting $\hat{{\varphi}}_{\hat{\boldsymbol{\beta}  },jj} = -1 $, we can define 
\begin{equation*}
  \hat{\boldsymbol{\Phi}}_{\hat{\boldsymbol{\beta}  }} =(-\hat{{\varphi}}_{\hat{\boldsymbol{\beta}  },ij})= 
\begin{pmatrix}
  1 & -\hat{{\varphi}}_{\hat{\boldsymbol{\beta}  },12} & \dots & -\hat{{\varphi}}_{\hat{\boldsymbol{\beta}  },1p}\\
  -\hat{{\varphi}}_{\hat{\boldsymbol{\beta}  },21} &1 & \dots & -\hat{{\varphi}}_{\hat{\boldsymbol{\beta}  },2p}\\
  \vdots & \vdots & \ddots & \vdots\\
  -\hat{{\varphi}}_{\hat{\boldsymbol{\beta}  },p1} &  -\hat{{\varphi}}_{\hat{\boldsymbol{\beta}  },p2} & \dots &1 
\end{pmatrix}.
\end{equation*}
Moreover, we define a diagonal matrix $
\hat{\mathbf{D}}_{\hat{\boldsymbol{\beta}  }}^2 = {\sf diag}( \hat{\phi}^2_{\hat{\boldsymbol{\beta}  },1},...,\hat{\phi}^2_{\hat{\boldsymbol{\beta}  },p} ),
$
with 
$\hat{\phi}^2_{\hat{\boldsymbol{\beta}  },j} =  \Vert \boldsymbol{X}_{\hat{\boldsymbol{\beta}  },(j)} - \mathbf{X}_{\hat{\boldsymbol{\beta}  },(-j)} \hat{\boldsymbol{\varphi}}_{\hat{\boldsymbol{\beta}  },j} \Vert^2_{2}/n + {\hat{\boldsymbol{\varphi}}^\top_{\hat{\boldsymbol{\beta}  },j} \nabla Q_{\lambda_{j}}(\hat{\boldsymbol{\varphi}}_{\hat{\boldsymbol{\beta}  },j})}$.
Then we can define the pseudo-inverse of $\hat{{\boldsymbol{\Sigma}}}_{\hat{\boldsymbol{\beta}  }} $ as $\hat{\boldsymbol{\Theta}}_{\hat{\boldsymbol{\beta}  }} $ with
\begin{equation*}
\hat{\boldsymbol{\Theta}}_{\hat{\boldsymbol{\beta}  }} = \hat{\mathbf{D}}_{\hat{\boldsymbol{\beta}  }}^{-2}\hat{\boldsymbol{\Phi}}_{\hat{\boldsymbol{\beta}  }}.
\end{equation*}
With the definitions, we can write  
$\hat{\boldsymbol{\Theta}}_{\hat{\boldsymbol{\beta}  },j}= \hat{\boldsymbol{\Phi}}_{\hat{\boldsymbol{\beta}  },j}/\hat{\phi}_{\hat{\boldsymbol{\beta}  },j}^2$  for $j\in\{1,\ldots, p\} $,  where $\hat{\boldsymbol{\Phi}}_{\hat{\boldsymbol{\beta}  },j}$ and $\hat{\boldsymbol{\Theta}}_{\hat{\boldsymbol{\beta}  },j} $
 are the $j$-th column of $\hat{\boldsymbol{\Phi}}^\top_{\hat{\boldsymbol{\beta}  }} $ and $\hat{\boldsymbol{\Theta}}^\top_{\hat{\boldsymbol{\beta}  }} $, respectively. 
 
 \begin{remark}
  Although $\hat{\boldsymbol{\Sigma}}_{\hat{\boldsymbol{\beta}}}$ is symmetric, its approximated inverse $\hat{\Theta}_{\hat{\boldsymbol{\beta}  }}$ is not guaranteed to be symmetric. 
  Moreover, the value of $\hat{\boldsymbol{\Theta}}_{\hat{\boldsymbol{\beta}  }} $ relies on $\hat{\boldsymbol{\beta}}$, which is different from the classic de-biased procedure under the OLS loss function and the classic high-dimensional precision matrix estimation problem, e.g., \cite{FHT08}, \cite{CLL11}, and \cite{LW17b}.
   
 \end{remark}


\subsection{Testing procedure}
We summary the testing procedure as follows, 
\begin{itemize}
    \item Step 1: Obtain the ALS estimator $ \hat{\boldsymbol{\beta} }$ by solving the optimization problem (\ref{est01}) under the expectile framework with an amenable regularizer. 
    \item Step 2: Estimate $\hat{\boldsymbol{\Theta}}_{\hat{\boldsymbol{\beta}  }}$,  the  pseudo-inverse of the generalized Hessian matrix $\hat{{\boldsymbol{\Sigma}}}_{\boldsymbol{\hat{\beta}}  }$,  by
    applying the node-wise regression method (\ref{op-de}) on the expectile-weighted design matrix $\mathbf{X}_{\hat{\boldsymbol{\beta} }}$ column by column. 
    \item Step 3: Obtain the de-biased estimator $\hat{\boldsymbol{\beta}}  _{de} $ by (\ref{debias-est}). 
 \item Step 4: Construct a Wald-type test,
\begin{equation}\label{eq14}
T_\mathbf{R}=(\mathbf{R}(\hat{\boldsymbol{\beta}}  _{de}-{\boldsymbol{\beta}}^*))^\top\hat{\boldsymbol{\Omega}}_\mathbf{R}^{-1}\mathbf{R}(\hat{\boldsymbol{\beta}}  _{de}-{\boldsymbol{\beta}}^*),
\end{equation}
for the hypothesis testing \eqref{test01}, where 
$
 \hat{\boldsymbol\Omega}_{\mathbf{R}} =\mathbf{R}\hat{\boldsymbol\Omega} \mathbf{R}^\top
$
 is the estimator of $
 {\boldsymbol\Omega}_{\mathbf{R}} =\mathbf{R}{\boldsymbol\Omega} \mathbf{R}^\top
$, the asymptotic variance of $\mathbf{R}(\hat{\boldsymbol{\beta}}  _{de}-{\boldsymbol{\beta}}^*)$ with
${\boldsymbol\Omega}$ being the asymptotic variance matrix of $\hat{\boldsymbol{\beta}}  _{de}$  defined by 
\begin{equation*}
{\boldsymbol\Omega}= (\omega_{ij}) =\boldsymbol{\Theta}_{\boldsymbol{\beta}^*} {\sf E} \left[ \boldsymbol{X}_i \boldsymbol{X}_i^\top w^4_{{\boldsymbol{\beta} ^*},i}\epsilon^2_i \right]{\boldsymbol{\Theta}}^\top_{\boldsymbol{\beta}^*},
\end{equation*}
and $ \hat{\boldsymbol\Omega}$ its estimator defined  by
 \begin{equation*}
 \hat{\boldsymbol\Omega} = (\hat\omega_{ij}) =\hat{\boldsymbol{\Theta}}_{\hat{\boldsymbol{\beta}}} \left(\frac{1}{n}\sum_{i=1}^n \boldsymbol{X}_i \boldsymbol{X}_i^\top w^4_{\hat{{\boldsymbol{\beta} }},i}\hat{\epsilon}^2_i \right)\hat{\boldsymbol{\Theta}}^\top_{\hat{\boldsymbol{\beta}}}.
 \end{equation*}
The test statistic asymptotically follows a $\chi^2$ distribution with the degree of freedom equal to the rank of $\mathbf{R}$, which we will show in Theorem \ref{thm3}.  Given a statistical significance level $\alpha$, the rejection region (which is a subset of the sample space) corresponding to the proposed test is 
\begin{equation*}
   \left\{\mathbf{X}: T_\mathbf{R}  >\chi^2_{\alpha}(p_0)\right\},
\end{equation*}
where $p_0$ is the rank of $\mathbf{R}$ and $\chi^2_{\alpha}(p_0)$  the upper $\alpha$-quantile of a $\chi^2$ distribution with $p_0$ degrees of freedom.
\end{itemize}






\section {Assumptions and statistical results. }\label{sec4}
In this section, we impose some technical conditions and establish statistical properties of the de-biased estimators. 

\subsection{Assumptions}


\begin{assumption}\label{assump1}
$(\boldsymbol{X}_i^\top, \epsilon_i)$ are i.i.d. random vectors with  ${\sf E}[w^2_{\boldsymbol{\beta}  ^*,i}\epsilon_i|\boldsymbol{X}_i] = 0$ and, for some positive constants $c_1, \ldots, c_6$,\\ 
(i) ${\sf E}[\epsilon_{i}^{4}]<c_1$, $\max_{i=1}^n|\epsilon_i
|=O_p(K)$, $\Pr \{|\epsilon_{i}|>c_2K\})=o((\ln p) /n)$, and {  $\sup_{x\in(-\infty,+\infty)}f_\epsilon(x)<\infty$,
where $f_\epsilon(\cdot) $ is the probability density function of $\epsilon$; } \\ 
 (ii) $\max_{j}{\sf E}[x_{ij}^{4}]<c_3$, $\Vert \mathbf{X}\Vert_{
 \infty}=O_p(K)$, $\Pr \{|x_{ij}|>c_4K\})=o((\ln p) /n)$;\\ 
(iii) $K^2\sqrt{(\ln p)/n}\to 0 $ as $n\to \infty$;\\
(iv)
$ 0 < c_5 < \lambda_{\min}({\boldsymbol{\Sigma}}
  ) < \lambda_{\max}({\boldsymbol{\Sigma}}) < c_6 < \infty $. 

\end{assumption}

Since we do not impose independence between $\boldsymbol{X}_{i} $ and $\epsilon_i$, the assumption allows for conditional   heteroscedastic cases.
The upper bound for $\Vert \mathbf{X}\Vert_{
 \infty}$ in Assumption 1(ii) is also considered in (D1) of \cite{VBRD14}, which is very general, as it include the case for bounded variables with $K=1$ and sub-Gaussian variables with $K=\ln (n \vee p)$. More generally, heavy tail cases may also included. For example, if $X_{ij}$ and $\epsilon_i$ have bounded $(4+\delta)$th moment, we can check using the Markov's inequality that the tail probability  assumption holds with $K=(np)^{1/(4+\delta)}$, and when  $p\asymp n^\eta$  it is easy to check that Assumption  1(iii),  also considered in (D2) of \cite{VBRD14}, holds for $\delta> 4\eta$. Assumption 1(iii) is considered in (D3) of \cite{VBRD14} for the identification reason.

\subsection{Preliminary theoretical results }
{  
In order to ensure the completeness of the paper, we start by stating the ${l}_2$ error bound, the ${l}_1$ error bound and the prediction error bound of the ALS estimators $\hat{\boldsymbol{\beta}}$. Similar results have been proved 
by \cite{GZ16} and \cite{LZZ22} under the sub-Gaussian settings.   Denoting by ${\cal A} = \{j; {\beta}^*_j \neq 0  \} $ the active set of the covariates and its cardinality  by $s := |{\cal A} | $, we have the following proposition. 
\begin{proposition}
\label{prop1}
  (The bounds for the initial estimators). Suppose that Assumption \ref{assump1} is satisfied. If $P_{\lambda}(\boldsymbol{\beta}  )$ is a $\mu$-amenable regularizer, $\min\{\tau,1-\tau\}\lambda_{\min}(\boldsymbol{\Sigma}) > 3\mu/4$, $\Vert \boldsymbol{\beta} \Vert_1 \leq R$, and $s\ln p/n = o(1)$, then by choosing $\lambda\geq c_7R\sqrt{\ln p/n}$ for some large positive constant $c_7$, the estimator given by the optimization (\ref{est01}) follows
 $$
  \Vert \hat{\boldsymbol{\beta}  } - \boldsymbol{\beta}  ^* \Vert_{2} = \mathcal{O}_{p}(\sqrt{s}\lambda),   \quad \Vert \hat{\boldsymbol{\beta}  } - \boldsymbol{\beta}  ^* \Vert_{1} = \mathcal{O}_{p}(s\lambda),   \quad \text{and} \quad \Vert \mathbf{X} (\hat{\boldsymbol{\beta}  } - \boldsymbol{\beta}  ^*) \Vert^2_{2}/n = \mathcal{O}_{p}(s\lambda^2). $$ 
\end{proposition}

\medskip
Next, we discuss the asymptotic property of the de-biased estimator. 
From (\ref{de-bias01}), we can see that
the weighting matrix $\mathbf{W}_{\hat{\boldsymbol{\beta}  }} $ is involved in both ${\boldsymbol{ \Delta}}^{(1)} $ and ${\boldsymbol{ \Delta}}^{(2)} $ as well as the construction of $\hat{\boldsymbol{\Theta}}_{\hat{\boldsymbol{\beta}  }} $. Thus a key step to derive the asymptotic theory is to show that the weighting matrix $\mathbf{W}_{\hat{\boldsymbol{\beta}  }} $ can be replaced by $\mathbf{W}_{\boldsymbol{\beta}  ^*} $ with little cost. 
Indeed, we have the following theorem.



\begin{lemma}\label{lem1}
 (The bounds of the weighting factor). 
  Under the assumptions of Proposition \ref{prop1}, we have 
  \begin{equation*}
  |w^2_{\boldsymbol{\beta} ^*,i} - w^2_{\hat{\boldsymbol{\beta}}  ,i}| \leq  \mathbb{I}\left(|\epsilon_i| \leq |\boldsymbol{X}_i^\top (\hat{\boldsymbol{\beta}}   - \boldsymbol{\beta} ^*) |\right),
  \end{equation*}
holds for $i\in\{1,\ldots,n\}$ and 
 \begin{equation}\label{eqT42.4}
  \frac{1}{n}\sum_{i=1}^n |w^2_{\hat{\boldsymbol{\beta}  },i} - w^2_{\boldsymbol{\beta}  ^*,i}|  =O_p\left( K s\lambda\right).
   \end{equation}
\end{lemma}
}

{  Lemma \ref{lem1} indicates that when $\hat{\boldsymbol{\beta}  } $ is close to $\boldsymbol{\beta}  ^* $, the weighting factor $w^2_{\hat{\boldsymbol{\beta}  },i}$ is equal to $w^2_{\boldsymbol{\beta}  ^*,i}$ with high probability, and the difference is equal to 1 with low probability. Therefore the average approximation error of the weighting factor converges in probability. This convergence rate differs from that in the  GLM scenario,  where the second-order derivative of the loss function holds the Lipschitz property, e.g., assumption (C1)  in \cite{VBRD14} and  (L1) in \cite{CGM21}.
Additionally, \cite{VBRD14} assume in (D4) that $|w_{\hat{\boldsymbol{\beta}  }, i} - w_{\boldsymbol{\beta}  ^*, i}| \leq |\boldsymbol{X}_i^\top (\hat{\boldsymbol{\beta}  } - \boldsymbol{\beta}  ^*)|$ and this directly implies that
$\frac{1}{n}\sum_{i=1}^n |w^2_{\hat{\boldsymbol{\beta}  }, i} - w^2_{\boldsymbol{\beta}  ^*, i}| \leq \Vert\mathbf{X}  (\hat{\boldsymbol{\beta}  } - \boldsymbol{\beta}  ^*)  \Vert^2_2/n = O_p(s\lambda^2)$, regardless of the condition on the maximum absolute value of the design matrix.  
Lemma \ref{lem1} shows that under the Assumption \ref{assump1} in expectile framework, the convergence rate is slower, which results in slower convergence rate in the subsequent node-wise regressions. 
}




\subsubsection{Node-wise regressions for the generalized Hessian matrix}





Denote the inverse of the population Hessian matrix of the weighted design $\mathbf{X}_{\boldsymbol{\beta}^*}$ by $\boldsymbol{\Theta}_{\boldsymbol{\beta}  ^*} = {\boldsymbol{\Sigma}}^{-1}_{\boldsymbol{\beta}  ^*} = ({\sf E} [\mathbf{X}^\top_{\boldsymbol{\beta}  ^*}\mathbf{X}_{\boldsymbol{\beta}  ^*}]/n)^{-1} $. For  $j\in\{1,\ldots, p\} $, denote the residual term of the node-wise regression by $ 
\boldsymbol{\varrho}_{\boldsymbol{\beta}  ^*,j} =X_{\boldsymbol{\beta}  ^*,(j)} -  X_{\boldsymbol{\beta}  ^*,(-j)} \boldsymbol{\varphi}_{\boldsymbol{\beta}  ^*,j},  
 $ where  $\boldsymbol{\varphi}_{\boldsymbol{\beta}  ^*,j} := \argmin_{\boldsymbol{\varphi}} {\sf E}\Vert X_{\boldsymbol{\beta}  ^*,(j)} -  \boldsymbol{X}_{\boldsymbol{\beta}  ^*,(-j)} \boldsymbol{\varphi}\Vert^2_2 $  is  the corresponding true coefficient. 
 Denote the corresponding population variance by $\phi^2_{\boldsymbol{\beta}  ^*,j} = {\sf E} \Vert \boldsymbol{\varrho}_{\boldsymbol{\beta}  ^*,j} \Vert^2_{2}/n.  $  

\begin{theorem}\label{thm1}
(The uniform bounds for the node-wise estimators). 
Suppose that all the assumptions in Proposition \ref{prop1} are satisfied. If $Q_{\lambda_{j}}(\boldsymbol{\varphi} )$ is $\mu-$amenable regularizers,  $\Vert \boldsymbol{\varphi}_{\boldsymbol{\beta} ^*,j} \Vert_1 \leq R_j$,  
$\max_j \Vert \boldsymbol{\varphi}_{\boldsymbol{\beta} ^*,j} \Vert_\infty=O(1)$, for $j\in\{1,\ldots, p\}$, $(s^{**})^2 K^2 \sqrt{\ln p / n } = o(1)$, and $s^{**}\lambda^{**} = o(1)$, where  $s^{**}=\max_j \Vert \boldsymbol{\varphi}_{\boldsymbol{\beta} ^*,j} \Vert_0$ and $\lambda^{**}=\max_j \lambda_j$,  then if we choose 
$\lambda_j\geq c_8(\max_jR_j\sqrt{\ln p/n})\vee(s\sqrt{s^{**}} K^3 \lambda))$
for some large positive constant $c_{8}$, we have (i)
 $$
\max_{1\leq j \leq p}\Vert \hat{\boldsymbol{\varphi}}_{\hat{\boldsymbol{\beta}  },j} - \boldsymbol{\varphi}_{\boldsymbol{\beta}  ^*,j} \Vert_{1} = \mathcal{O}_{p}(s^{**}\lambda^{**}), 
\quad    \max_{1\leq j \leq p}\Vert \hat{\boldsymbol{\varphi}}_{\hat{\boldsymbol{\beta}  },j} - \boldsymbol{\varphi}_{\boldsymbol{\beta}  ^*,j} \Vert_{2} = \mathcal{O}_{p}(\sqrt{s^{**}}\lambda^{**}),$$
and 
$$ \max_{1\leq j \leq p}\Vert \mathbf{X}_{\boldsymbol{\beta}  ^*,(-j)} (\hat{\boldsymbol{\varphi}}_{\hat{\boldsymbol{\beta}  },j} - \boldsymbol{\varphi}_{\boldsymbol{\beta}  ^*,j} )\Vert^2_{2}/n = \mathcal{O}_{p}(s^{**}(\lambda^{**})^2).$$
(ii) 
 $$ \max_{1\leq j \leq p}\left|\hat{\phi}_{\hat{\boldsymbol{\beta}  },j}^2-\phi_ {\boldsymbol{\beta}  ^*,j}^2\right| = O_p(\sqrt{s^{**}}\lambda^{**}), $$   
 $$\max_{1\leq j \leq p}\Vert \hat{\boldsymbol{\Theta}}_{\hat{\boldsymbol{\beta}  },j} - \boldsymbol{\Theta}_{\boldsymbol{\beta}  ^*,j} \Vert_{1} = O_p({s^{**}}\lambda^{**}), \quad \textit{and} \quad   \max_{1\leq j \leq p}\Vert \hat{\boldsymbol{\Theta}}_{\hat{\boldsymbol{\beta}  },j} - \boldsymbol{\Theta}_{\boldsymbol{\beta}  ^*,j}  \Vert _{2} = O_p(\sqrt{s^{**}}\lambda^{**}).
 $$
\end{theorem}


The theorem provides uniformly convergence of the node-wise regression estimators. When the conditions $K \asymp O(1)$ and $s \asymp s^{**} \asymp O(1)$ are simultaneously satisfied, we can choose $R \asymp \max_{j}R_j \asymp O(1)$ and $\lambda \asymp \lambda_j\asymp  O(\sqrt{\ln p / n })$, for $j\in\{1,\ldots, p\}$. In such a case, the  $l_1 $ and $l_2 $ bounds for both  $\hat{\boldsymbol{\beta}  } $ and $\hat{\boldsymbol{\varphi}}_{\hat{\boldsymbol{\beta}  },j} $ have reached their optimal bounds  under the regularized framework, see \cite{CG17}, \cite{STZ12} and  \cite{Ver12}. 
Compared to the results in Theorem 3.2 of \cite{VBRD14} which choose  $\lambda_j  \asymp O( K \sqrt{\ln n/p} )$,
we may choose $\lambda_j  \asymp O( RR_js\sqrt{s^{**}}K^3 \sqrt{\ln n/p} )$.  { The primary cause of this difference is the second-order non-Lipschitz property of the expectile loss function. Additionally, the use of the non-convex regularizer and the uniform rate under consideration also contribute to this disparity.} For the same reason, the convergence rate is slower than that in Thoerem 3.2 of \cite{VBRD14}.



\subsection{Main results}

We establish the probabilistic upper bound of ${\boldsymbol{ \Delta}}^{(1)}$ and ${\boldsymbol{ \Delta}}^{(2)}$ as defined in \eqref{de-bias01} and  the asymptotic normality of the de-bias estimator under the expectile framework in the following two theorems, respectively.

\begin{theorem}\label{thm_add}
  (Square-root $n$ negligiblity of ${\boldsymbol{ \Delta}}^{(1)} $  and ${\boldsymbol{ \Delta}}^{(2)} $). Suppose that all the conditions  in Theorem \ref{thm1} hold and additionally we assume that $s^{**}\lambda\lambda^{**} = o(n^{-1/2})$, then  
 $$
\Vert {\boldsymbol{ \Delta}}^{(1)} \Vert_{\infty} = o_p(n^{-1/2}), \quad \text{and} \quad  \Vert {\boldsymbol{ \Delta}}^{(2)} \Vert_{\infty} = o_p(n^{-1/2}).
 $$


\end{theorem}



\begin{theorem}\label{thm3}
  (Asymptotic normality for the de-biased estimator).
  Suppose that all the conditions  in Theorem \ref{thm_add} hold, then  for the $p_0
$-by-$p$ hypothesis matrix $\mathbf{R}$ defined in \eqref{test01}  with  $\Vert\mathbf{R}\Vert_\infty=O(1)$, $\Vert\mathbf{R}\Vert_{l_\infty}=O(1)$ and $\mathbf{R}^\top\mathbf{R}\to \boldsymbol{\Upsilon}$ where $\boldsymbol{\Upsilon}$ is a  full row-rank $p_0
$-by-$p_0$ matrix, we have
  $$
\sqrt{n}{\mathbf R}\left(\hat{\boldsymbol{\beta}}_{de} - \boldsymbol{\beta}^*\right)  \xrightarrow{d}\mathcal{N}_{p_0}(\boldsymbol{0},\boldsymbol{\Omega}_{\mathbf R}),
 $$
{ and moreover if ${\sf E}[\epsilon_{i}^8]<c_1$,  $\max_{j}{\sf E}[x_{ij}^8]<c_3$, $K^4 \sqrt{\ln p / n } = o(1)$,
and $(s^{**})^{3/2}\lambda^{**} = o(1)$, we have}
 $$\Vert\hat{\boldsymbol\Omega}_{\mathbf{R}} -{\boldsymbol\Omega}_{\mathbf{R}}\Vert_\infty =o_p(1),$$
 where ${\boldsymbol\Omega}_{\mathbf{R}}$ and $\hat{\boldsymbol\Omega}_{\mathbf{R}}$ are defined in \eqref{eq14}. Consequently, we have 
\begin{equation*}
(\mathbf{R}(\hat{\boldsymbol{\beta}}  _{de}-{\boldsymbol{\beta}}^*))^\top\hat{\boldsymbol{\Omega}}_\mathbf{R}^{-1}\mathbf{R}(\hat{\boldsymbol{\beta}}  _{de}-{\boldsymbol{\beta}}^*)\xrightarrow{d} \chi^2(p_0).
\end{equation*}
\end{theorem}
Compared to the Theorem 2.4 in \cite{VBRD14}, the additional $8$th moment conditions on $\epsilon_i$  and $x_{ij}$ are essential to guarantee the consistency of $\hat{\boldsymbol\Omega}_{\mathbf{R}} $.
The assumption $K^4\sqrt{(\ln p)/n}\to 0$ is valid even when 
 $\boldsymbol{X}_{ij}$ and $\epsilon_i$ do not necessarily follow sub-Gaussian or strongly bounded distributions. For instance,   when $s$ and $s^{**}$ are fixed positive integers, $p\asymp n^\eta$,  and $X_{ij}$ and $\epsilon_i$ have bounded $(8+\delta)$th moment for some $\delta> 8\eta$, it is easy to check that $K^4\sqrt{(\ln p)/n}\to 0$ holds.
 {  Although various regularizers within the amenable category yield similar performance in terms of the 
$l_1$ and $l_2$  bounds when estimating corresponding coefficients, non-convex regularizers offer distinct advantages in variable selection. This is further elaborated in \cite{LW15}, \cite{ZCZ18}, among others.} 
By reducing falsely selected coordinates in both $ \hat{\boldsymbol{\beta}  } $ and $ \hat{\boldsymbol{\varphi}}^2_{\hat{\boldsymbol{\beta}  },j}$ in the estimation of the pseudo-inverse of the sample covariance matrix,  our proposed test with non-convex regularizers may achieve extra efficacy.{  We show in the simulation that the test constructed using the SCAD regularizer performs better than that constructed using the LASSO regularizer in most situations, especially for the heteroscedasticity case.}

\section{Simulation study}\label{sec5}

In this section, we conduct simulation studies to evaluate the finite sample performance of the proposed Wald-type test under the expectile framework. Our primary focus is on assessing its ability to control Type I error and analyzing its local power in various scenarios. Meanwhile, recall that there are totally two regularizers, $P_{\lambda}(\cdot) $ and $Q_{\lambda_j}(\cdot) $, separately used in our testing procedure. We also exhibit a strong interest in the performance of test statistics derived from different regularizers within the amenable category. To illustrate this, we exemplify the convex regularizer with the Lasso and the non-convex regularizer with the SCAD and denote the combination of two LASSO regularizers as 'LA-LA' and  the combination of two SCAD regularizers as 'SC-SC'.

Given that expectile regression (\ref{est01}) can be applied to detect heteroscedasticity in high-dimensional data due to its asymmetric weighting on squared loss, we consider both the homoscedastic models studied in \cite{VBRD14} and  heteroscedastic models discussed in \cite{GZ16}, \cite{WWL12}, and \cite{ZCZ18}.
The entire simulation study is conducted using {\sf MATLAB 2022b}, with each simulation procedure repeated 1,000 times.  The tuning parameters $\lambda $ and $\lambda_j $ are selected via the $10- $fold cross-validation. Additionally, we set the crucial parameter $\gamma = 3.7 $ for the SCAD regularizer, as  suggested in \cite{FL01}.

\subsection{Simulation results under the homoscedastic case}

The response variable is generated from the following linear expectile model:
\begin{equation*}
y_i = \boldsymbol{X}_i^\top \boldsymbol{\beta} ^*  + (\epsilon_i - {\sf E}\rho_{\tau}(\epsilon_i)), 
\end{equation*}  
for $i=1,\ldots,n$, where the error term $\epsilon_i - {\sf E}\rho_{\tau}(\epsilon_i)$ ensures that the identical condition is satisfied. 

 To generate the covariates $\boldsymbol{X}_i$, we draw them independently from the multivariate normal distribution $\mathcal{N}_p(0,{\boldsymbol{\Sigma}}) $. We specify  two distinct covariance matrices  as the design choices, which are studied by \cite{VBRD14} and \cite{LW17b}, respectively: 
\begin{align*}
&\text{Toeplitz: } \quad   {\boldsymbol{\Sigma}}_{jk} = \xi^{|j-k|},  \\
&\text{Scale-free graph corr:} \quad 
{\boldsymbol{\Sigma}} = \left[ {\mathbf D} \left[ {\mathbf A} + (|\lambda_{min}({\mathbf A})| + 0.2) {\mathbf I}\right]{\mathbf D} \right]^{-1},
\end{align*}
where  ${\mathbf{A}} \in {\mathbb{R}}^{p \times p}$ is an adjacency matrix associated with a certain graph. In this matrix, the nonzero off-diagonal elements ${{A}}_{jk}$, $|j-k| \leq \varsigma $  are set to $0.3$, and the diagonal elements are set to $0$. Moreover, ${\mathbf{D}} \in {\mathbb{R}}^{p \times p}$ is a diagonal matrix where ${{D}}_{jj} = 1$ for $j\in\{1,\ldots, p/2\}$ and ${{D}}_{jj} = 3$ for $j\in\{p/2 + 1,\ldots, p\}$. Notably,   the Toeplitz covariance matrix leads to tridiagonal precision inverse matrix, while the sparsity of the inverse matrix of the Scale-free graph corr depends on the adjacency relationships recorded by ${\mathbf A}$.

To generate $\boldsymbol{\beta} ^*$, we let ${\beta}^*_1 = k/\sqrt{n}$,  $k\in\{0,\ldots, 6\}$ in different settings. For the rest of the coefficients, we consider two scenarios: 
\begin{enumerate}[(1)]
\item The Dirac measure case: ${\beta}^*_i = 1$ if $i\in {\cal K}$ and ${\beta}^*_i = 0$ if $i\notin {\cal K}\cup\{1\}$. 
\item The Uniform random measure case:  ${\beta}^*_i\sim\mathcal{U}(0,2)$ if $i\in {\cal K}$ and ${\beta}^*_i = 0$ if $i\notin {\cal K}\cup\{1\}$.
\end{enumerate}
where ${\cal K}={\cal K}^{4}$ or ${\cal K}^{10}$, with 
 $$
{\cal K}^{4} \in \{6,12,15,20  \}, \quad \text{and} \quad {\cal K}^{10} \in \{5,6,7,8,9,10,11,12,15,20  \}.
 $$
We denote the two scenarios as Dirac 4 (Dirac 10) and Unif 4 (Unif 10), respectively.

To generate $\epsilon_i$, 
 two distributions are taken into consideration: (1) the  standard normal distribution $\mathcal{N}(0,1)$, and (2) the student-$t$ distribution with $4$ degrees of freedom, denoted as $t_4$. The $t_4$ distribution is a heavy-tailed distribution, which helps to evaluate the performance of our proposed test under heavy-tailed scenario.

The sample size is set as $n =300$, the dimension of the covariates is set as $p = 200,400$, and $600$  and  the expectile level is set to $\tau = 0.1,0.5$ and $0.9$. Note that when the expectile level $\tau = 0.5$, the asymmetric square loss becomes exactly equivalent to the least square loss. The corresponding de-biased estimator then aligns with that under the OLS framework, as studied in \cite{VBRD14, JM14} and \cite{ZZ14}. For the parameters in different choices of the covariance matrix, we consider $\xi = 0.25,0.5$, and $0.75$ for the Toeplitz while $\varsigma = 10, 20$, and $p$ for the Scale-free graph case.

To assess the ability of the test statistic in controlling the Type I error and local power in homoscedastic scenarios, we firstly conduct a hypothesis test on a single coefficient at given expectile level $\tau$,
\begin{equation}\label{sig_test_Homo}
  {\sf H}_{\tau,0}: {\beta}^*_1 = 0 \quad \text{versus} \quad {\sf H}_{\tau,1}: {\beta}^*_1 \neq 0.
\end{equation} 
\begin{table}[htbp]
\small
  \centering
  \caption{The empirical Type I error and power  obtained by different regularizers under the homoscedastic case with  $n = 300, p =400$,  Dirac $4 $, $\tau = 0.1 $, standard normal error and Toeplitz design  from $1000 $ replicates, and  the average computing time for each repetition. The CPU time records the average computation time for each repetition under different methods. The null hypothesis is ${\sf H}_0: \beta^*_1=0$.}
    \begin{tabular}{cccccccccc}
    \toprule
    \toprule
    Toeplitz  &  $\epsilon \sim \mathcal{N}(0,1)$    & \multicolumn{7}{c}{$\beta^*_1 = k/\sqrt{n}$} &  \\
\cmidrule{3-9}    Method &  & $k=0$   & $k=1$   & $k=2$   &$ k=3$   & $k=4$   & $k=5$   &$ k=6$   & CPU time (s) \\
    \midrule
          & $\xi=0.25$ & 5.90\% & 13.50\% & 35.80\% & 65.40\% & 87.10\% & 96.80\% & 99.50\% &  \\
    LA-LA & $\xi= 0.50$  & 5.20\% & 11.80\% & 35.00\% & 64.20\% & 85.20\% & 95.60\% & 99.10\% & 1.056  \\
          & $\xi = 0.75$ & 4.50\% & 11.90\% & 28.10\% & 52.30\% & 74.20\% & 87.50\% & 95.70\% &  \\
    \midrule
          & $\xi=0.25$ & 4.80\% & 12.30\% & 35.30\% & 68.80\% & 89.00\% & 95.30\% & 99.60\% &  \\
    SC-SC & $\xi =  0.50$  & 4.60\% & 12.50\% & 33.10\% & 67.20\% & 86.70\% & 94.10\% & 99.20\% & 4.252  \\
          & $\xi = 0.75$ & 4.90\% & 10.80\% & 26.00\% & 59.30\% & 88.50\% & 89.10\% & 95.20\% &  \\
    \bottomrule
    \bottomrule
    \end{tabular}%
  \label{tab:addlabel}%
  \label{diff_01table}
\end{table}%

\begin{figure}[htbp]
	\centering
	\begin{minipage}[t]{0.48\textwidth}
		\centering
ude		\includegraphics[width=7cm]{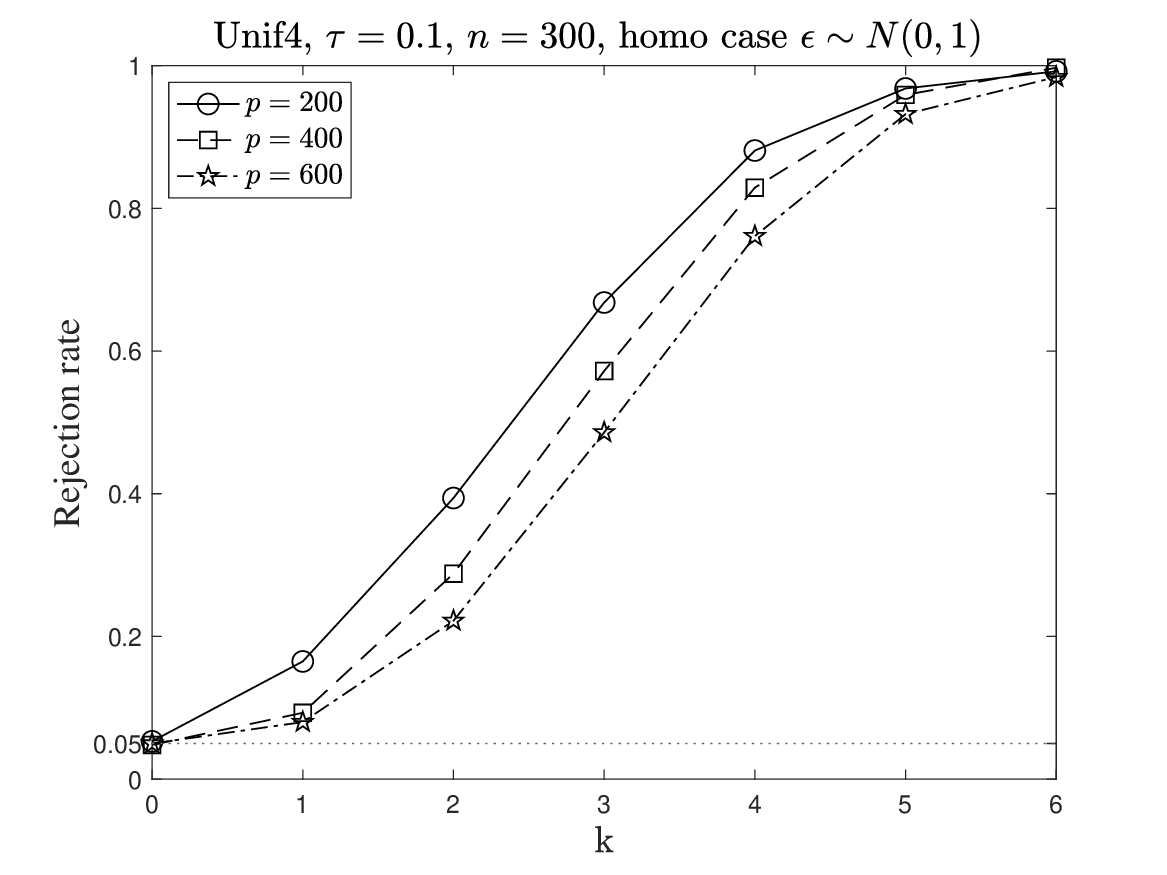}
	\end{minipage}
	\begin{minipage}[t]{0.48\textwidth}
		\centering
		\includegraphics[width=7cm]{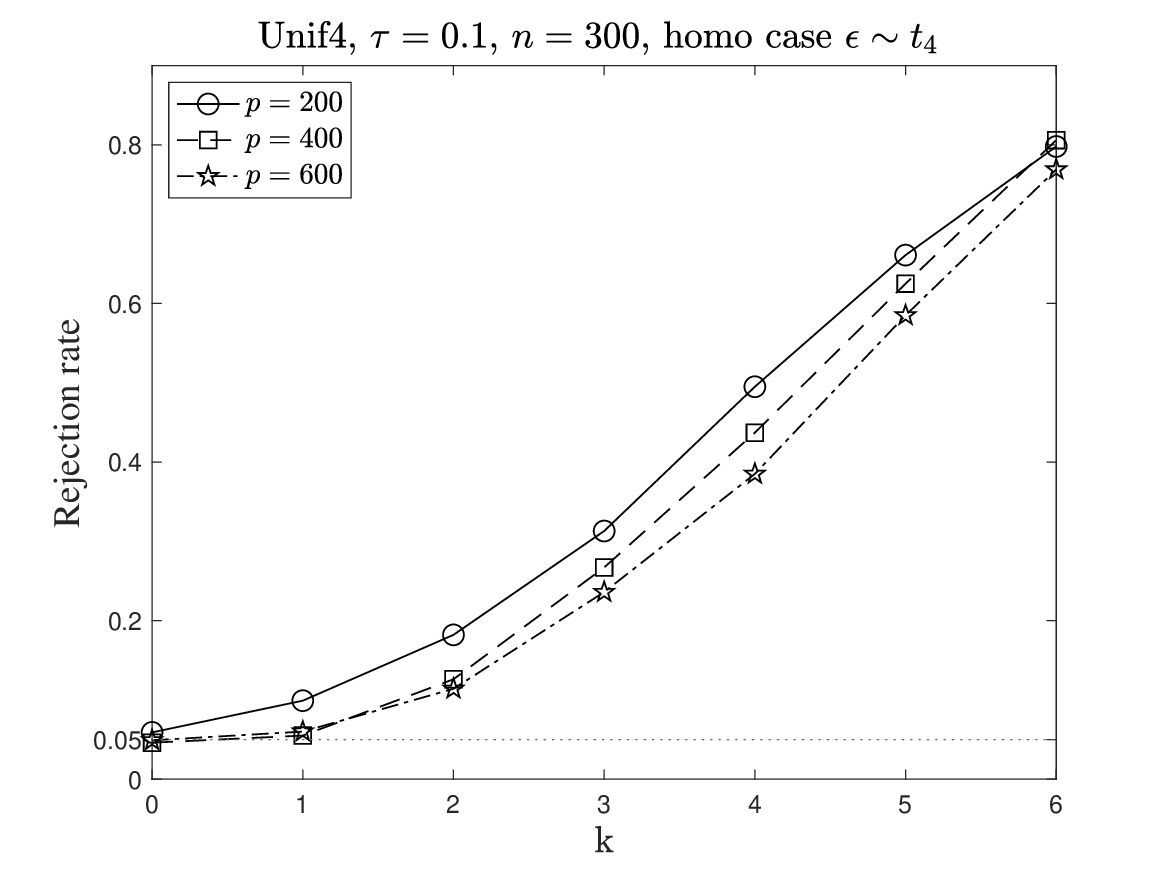}
	\end{minipage}
  \caption{The empirical power and Type I error of the proposed de-biased test for the expectile regression with $\tau = 0.1$ and $\beta^*_1 = k/\sqrt{n}$ under the Unif4, homoscedastic case with $n = 300$ and $p\in\{200,400,600\}$, calculated from $1000 $ replicates. The null hypothesis is ${\sf H}_0: \beta^*_1=0$.}
  \label{pic_300_varyp_homo}
\end{figure}

\begin{figure}[htbp]
	\centering
	\begin{minipage}[t]{0.48\textwidth}
		\centering
		\includegraphics[width=7.5cm]{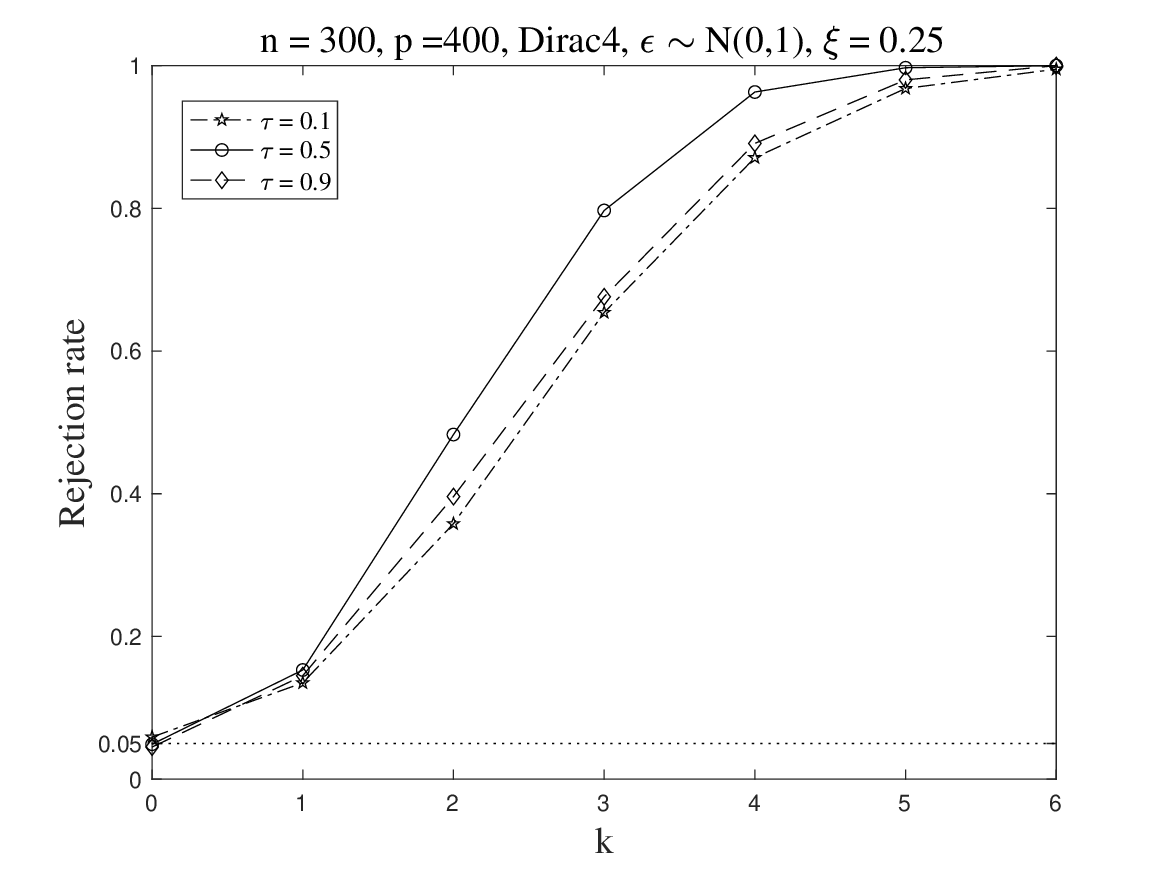}
	\end{minipage}
	\begin{minipage}[t]{0.48\textwidth}
		\centering
		\includegraphics[width=7.5cm]{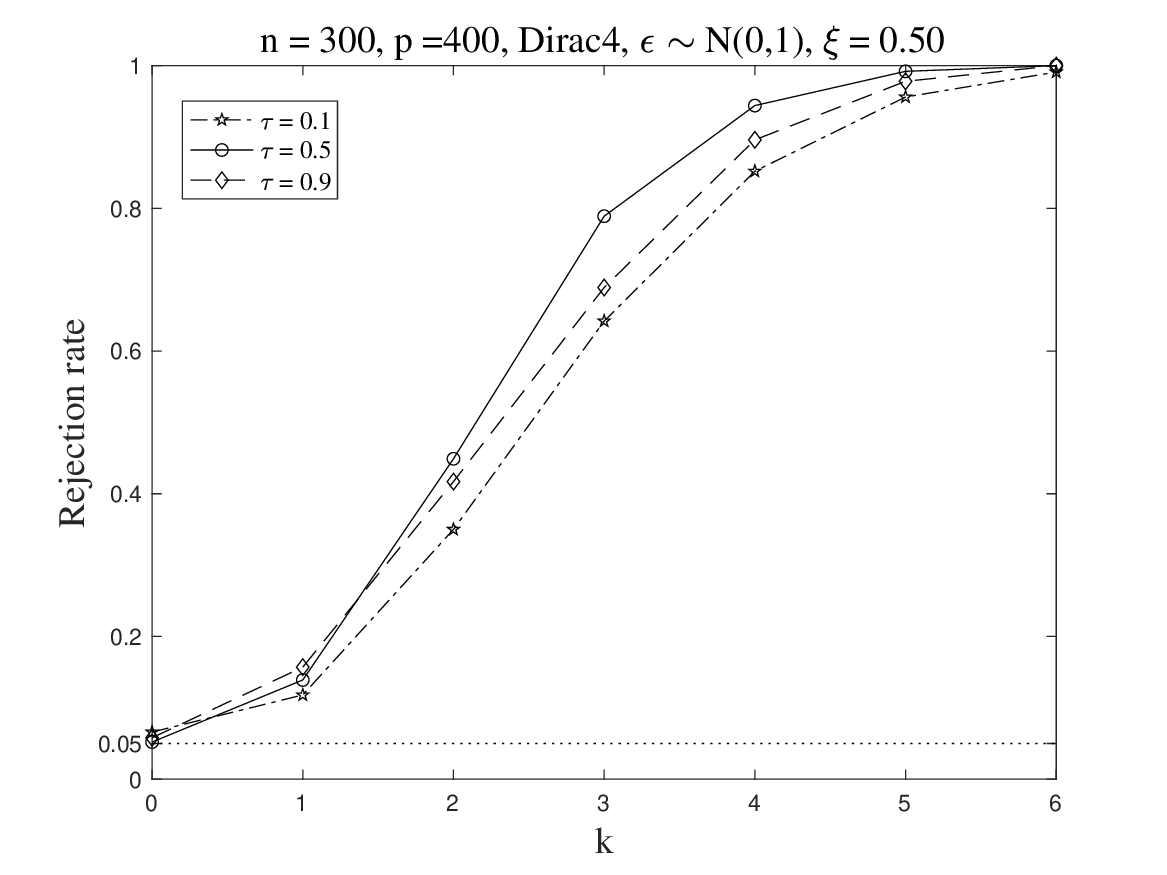}
	\end{minipage}
  \caption{The empirical power and Type I error of the proposed de-biased test for the expectile regression  under the Dirac 4,  homoscedastic case with $n = 300$, $\beta^*_1 = k/\sqrt{n}$,  $\tau\in\{0.1,0.5,0.9\}$ and Toeplitz design  calculated from $1000 $ replicates. The null hypothesis is ${\sf H}_0: \beta^*_1=0$.}
  \label{pic_300_vartau_homo}
\end{figure}

Table \ref{diff_01table} shows that test statistics derived from different regularizers exhibit similar efficacy in controlling Type I error and local power under the Dirac 4 scenario with diverse Toeplitz designs. However, the use of the non-convex SCAD regularizer incurs a noticeably higher computational burden, even when implementing the LLA strategy by {\cite{ZL08}}, than the Lasso method. 
 Figure \ref{pic_300_varyp_homo} shows that as the dimension grows, the empirical power of the proposed test gradually diminishes. Furthermore, in scenarios where errors follow a 
t-distribution, the test's empirical power is observed to be relatively lower compared to scenarios with normally distributed errors.
Figure \ref{pic_300_vartau_homo} demonstrates that the rejection rate of our proposed test is comparable at $\tau = 0.1$ and $\tau = 0.9$. Notably, the test exhibits marginally higher empirical power at $\tau = 0.5$, which corresponds to the OLS test in \cite{VBRD14},  in comparison to the other two expectile levels.

Following the results presented in Table \ref{diff_02table}, we observe that the two methods exhibit similar test power at $\varsigma = 10$. Notably, as $\varsigma$ increases from 20 to 
$p$, the 'SC-SC' method outperforms 'LA-LA' in controlling the empirical Type I error along with local power. This suggests that the non-convex penalizer might be more effective in scenarios where the inverse of the covariance matrix is not particularly sparse.

Next we consider a group test on ${ \boldsymbol\beta}^*_{\cal G} $, specifically,
\begin{equation}\label{sig_gptest_Homo}
  {\sf H}_{0}: { \boldsymbol\beta}^*_{\cal G} = \boldsymbol{0} \quad \text{versus} \quad {\sf H}_{1}: { \boldsymbol\beta}^*_{\cal G}   \neq \boldsymbol{0},
\end{equation}
where ${\cal G} = \{ 1,3,4 \}$. Analogously,  we set the value of $\beta^*_1$ varies from $0/\sqrt{n}$ to $6/\sqrt{n}$  while we should note that $\beta^*_3$ and $\beta^*_4$ are consistently zero. Table \ref{homo_group_test} reports the empirical Type I error and  power for  multiple tests on  group ${\cal G}$ across various settings on the true coefficients, which  is consistent with the result on the test of single component  $\beta^*_1$ under various coefficient setups in our simulation study.

\begin{table}[htbp]
\small
  \centering
  \caption{The empirical Type I error and power  obtained by different regularizers under the  Dirac $4 $, homoscedastic case with $n = 300, p =400$,   $\tau = 0.1 $, standard normal error and Scale-free design  from $1000 $ replicates. The null hypothesis is ${\sf H}_0: \beta^*_1=0$.}
    \begin{tabular}{ccccccccc}
    \toprule
    \toprule
    Scale-free  &   $\epsilon \sim \mathcal{N}(0,1)$    & \multicolumn{7}{c}{$\beta^*_1 = k/\sqrt{n}$} \\
\cmidrule{3-9}    Method &  & $k=0$   & $k=1$   & $k=2$   & $k=3$   & $k=4$   & $k=5$   & $k=6$ \\
    \midrule
          & $\varsigma = 10$ & 5.10\% & 7.10\% & 12.10\% & 30.10\% & 49.60\% & 72.10\% & 86.70\% \\
    LA-LA & $\varsigma = 20$ & 5.30\% & 6.90\% & 7.70\% & 16.50\% & 31.70\% & 61.30\% & 76.70\% \\
          & $\varsigma = 400$ & 7.50\% & 47.50\% & 93.80\% & 99.90\% & 100.00\% & 100.00\% & 100.00\% \\
    \midrule
          & $\varsigma = 10$ & 4.70\% & 7.20\% & 19.70\% & 41.30\% & 64.00\% & 88.30\% & 92.30\% \\
    SC-SC & $\varsigma = 20$ & 5.00\% & 7.30\% & 16.10\% & 29.40\% & 46.60\% & 66.80\% & 81.20\% \\
          & $\varsigma = 400$ & 4.80\% & 45.10\% & 86.20\% & 99.30\% & 100.00\% & 100.00\% & 100.00\% \\
    \bottomrule
    \bottomrule
    \end{tabular}%
 \label{diff_02table}
\end{table}%

\begin{table}[htbp]
\small
  \centering
  \caption{The empirical Type I error and power for group ${\cal G}$ under the homoscedastic case with $n = 300, p =400$,  $\xi = 0.50$, $\tau = 0.1 $, standard normal error and Toeplitz design  from $1000 $ replicates. The null hypothesis is ${\sf H}_0: \beta^*_1=\beta^*_3=\beta^*_4=0$.}
    \begin{tabular}{lrrrrrrr}
    \toprule
    \toprule
     & \multicolumn{7}{c}{$\beta^*_1 = k/\sqrt{n},\  \beta^*_3=\beta^*_4=0, \quad {\sf H}_0: \beta^*_1=\beta^*_3=\beta^*_4=0$.} \\
\cmidrule{2-8}     & \multicolumn{1}{l}{$k = 0$} & \multicolumn{1}{l}{$k = 1$} & \multicolumn{1}{l}{$k = 2$} & \multicolumn{1}{l}{$k = 3$} & \multicolumn{1}{l}{$k = 4$} & \multicolumn{1}{l}{$k = 5$} & \multicolumn{1}{l}{$k = 6$} \\
    \midrule
    Dirac4 & 5.30\% & 14.70\% & 30.90\% & 57.40\% & 80.80\% & 93.50\% & 98.90\% \\
    Dirac10 & 6.30\% & 14.00\% & 30.30\% & 55.10\% & 77.00\% & 91.50\% & 97.60\% \\
    Unif4 & 4.70\% & 13.20\% & 32.60\% & 59.40\% & 81.70\% & 95.10\% & 99.00\% \\
    Unif10 & 6.50\% & 15.70\% & 31.70\% & 56.20\% & 78.50\% & 92.00\% & 99.30\% \\
    \bottomrule
    \bottomrule
    \end{tabular}%
  \label{homo_group_test}%
\end{table}%

\subsection{Simulation results under the heteroscedastic case}

The response variable is generated from the following sparse linear model:
\begin{equation*}
 y_i = x_{i,6} + x_{i,12} + x_{i,15} + x_{i,20} + 0.7 \Phi(x_{i,1}){\epsilon_i},
\end{equation*}  
for $i \in \{1,\ldots,n\}$,  where the covariate $\boldsymbol{X}_i$ being a $p$-dimensional random variable and the error term $\epsilon_i$ being  generated as in the homoscedastic case. Additionally,   $\Phi(\cdot) $ denotes the cumulative distribution function (cdf) of the standard normal distribution. 
Note further that the covariate $x_{(1)} $  plays an essential role in revealing the potential heteroscedasticity, since the magnitude of the  pseudo-true coefficient for $x_{(1)} $ by the means of expectation under expectile framework is strictly distinct from $0 $ for any  $\tau \neq 0.5 $. Thus, many studies associate the presence of heteroscedasticity with a non-zero estimate of the coefficient corresponding to $x_{(1)} $, see  \cite{GZ16} and \cite{ZCZ18}.  

 To evaluate the capability of our proposed test  in detecting heteroscedasticity, we continue to employ the  test  (\ref{sig_test_Homo})  for single coordinate to obtain the rejection rates under a given statistical significance level $\alpha$ at different expectile levels. 
 Furthermore, we are intrigued by the significance of covariates  that exhibit high correlation with $x_{(1)}$ and, as a result, may be influenced by such heteroscedasticity. Consequently, we also conduct an additional test to further investigate the performance of our proposed test in controlling  the Type I error and local power,
 \begin{equation*}\label{sig_test_ad}
  {\sf H}_{\tau,0} : {\beta}^*_2 = 0 \quad \text{and} \quad {\sf H}_{\tau,1}:{\beta}^*_2 \neq 0,  
\end{equation*}
where ${\beta}^*_2 $  is the coefficient of $x_{(2)} $ in the heteroscedastic model. 

\begin{table}[htbp]
\small
  \centering
  \caption{The empirical rejection rate of the test (\ref{sig_test_Homo}) deduced by our proposed test statistic with 'LA-LA' method and 'SC-SC' method, and the frequency that $x_{(1)} $ is selected by the 'E-SCAD' and 'E-Lasso' method under the heteroscedastic case with $n = 300$, $\xi = 0.5$, Toeplitz design,  different expectile level $\tau $ and error distribution from 1000 replicates.}
    \begin{tabular}{ccccccccc}
    \toprule
    \toprule
    Dirac4 & \multicolumn{4}{c}{$\epsilon \sim \mathcal{N}(0,1)$}        & \multicolumn{4}{c}{$\epsilon \sim t_4 $} \\
\cmidrule{2-9}    $p = 400$ & LA-LA & E-Lasso & SC-SC & E-Scad & LA-LA & E-Lasso & SC-SC & E-Scad \\
    \midrule
    $\tau = 0.1$ & 90.50\% & 85.60\% & 98.20\% & 86.10\% & 82.40\% & 83.10\% & 98.30\% & 74.60\% \\
    $\tau = 0.5$ & 5.40\% & 5.70\% & 5.10\% & 0.00\% & 5.00\% & 3.30\% & 4.70\% & 0.00\% \\
    $\tau = 0.9$ & 89.70\% & 91.00\% & 95.10\% & 93.80\% & 79.70\% & 88.90\% & 93.20\% & 78.60\% \\
    \midrule
    Dirac4 & \multicolumn{4}{c}{$\epsilon \sim \mathcal{N}(0,1)$}        & \multicolumn{4}{c}{$\epsilon \sim t_4$} \\
\cmidrule{2-9}    $p = 600$ & LA-LA & E-Lasso & SC-SC & E-Scad & LA-LA & E-Lasso & SC-SC & E-Scad \\
    \midrule
    $\tau = 0.1$ & 89.70\% & 83.70\% & 96.60\% & 79.50\% & 74.90\% & 68.80\% & 96.10\% & 69.10\% \\
    $\tau = 0.5$ & 4.80\% & 1.60\% & 5.00\% & 0.00\% & 4.70\% & 1.80\% & 4.70\% & 0.00\% \\
    $\tau = 0.9$ & 87.70\% & 78.60\% & 94.40\% & 87.30\% & 80.70\% & 69.20\% & 96.50\% & 72.20\% \\
    \bottomrule
    \bottomrule
    \end{tabular}%
  \label{hete_table_02}%
\end{table}%


Table \ref{hete_table_02} presents the performance of our proposed test statistic, along with the expectile-based estimators 'E-SCAD' and 'E-Lasso', which have been studied in \cite{ZCZ18}. In the case of $\tau =0.5$, it is observed that the 'E-Lasso' estimator tends to select $x_{(1)} $ more frequently, while the 'E-SCAD' estimator is less inclined to select $x_{(1)} $. Furthermore,  our proposed test effectively controls the empirical Type I error in cases where $p =400$ and $p = 600$. For the case $\tau = 0.1 $ and $\tau = 0.9 $, the empirical rejection rate of our proposed test is larger than the frequency that $x_{(1)} $ is selected by the 'E-SCAD' and smaller than that selected by the 'E-Lasso', which demonstrates that our test can relatively well identify the heteroscedasticity of the model. Additionally, under the Dirac 4 Toeplitz design, the test statistics derived using the 'SC-SC' method exhibit relatively superior capability in detecting heteroscedasticity compared with those derived from the 'LA-LA' method.

 \begin{figure}[htbp]
	\centering
	\begin{minipage}[t]{0.48\textwidth}
		\centering
		\includegraphics[width=7.5cm]{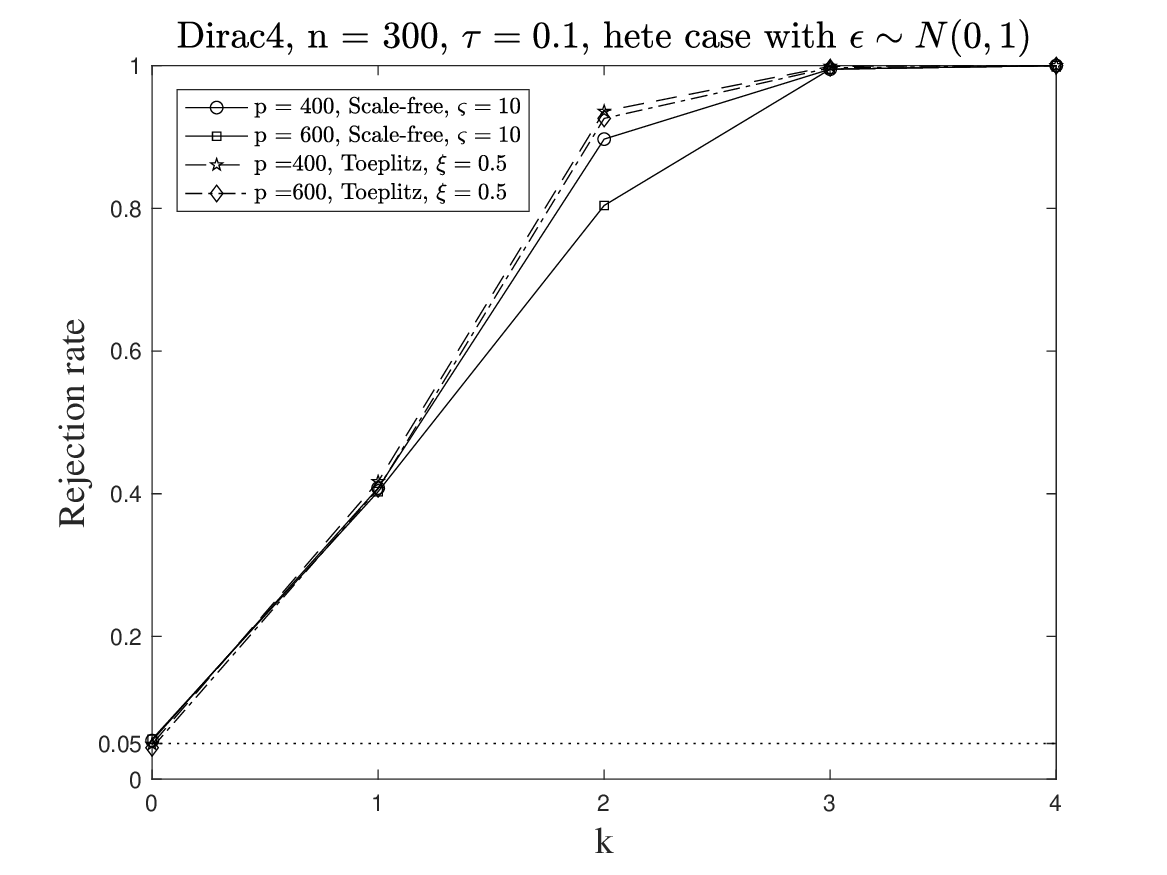}
	\end{minipage}
	\begin{minipage}[t]{0.48\textwidth}
		\centering
		\includegraphics[width=7.5cm]{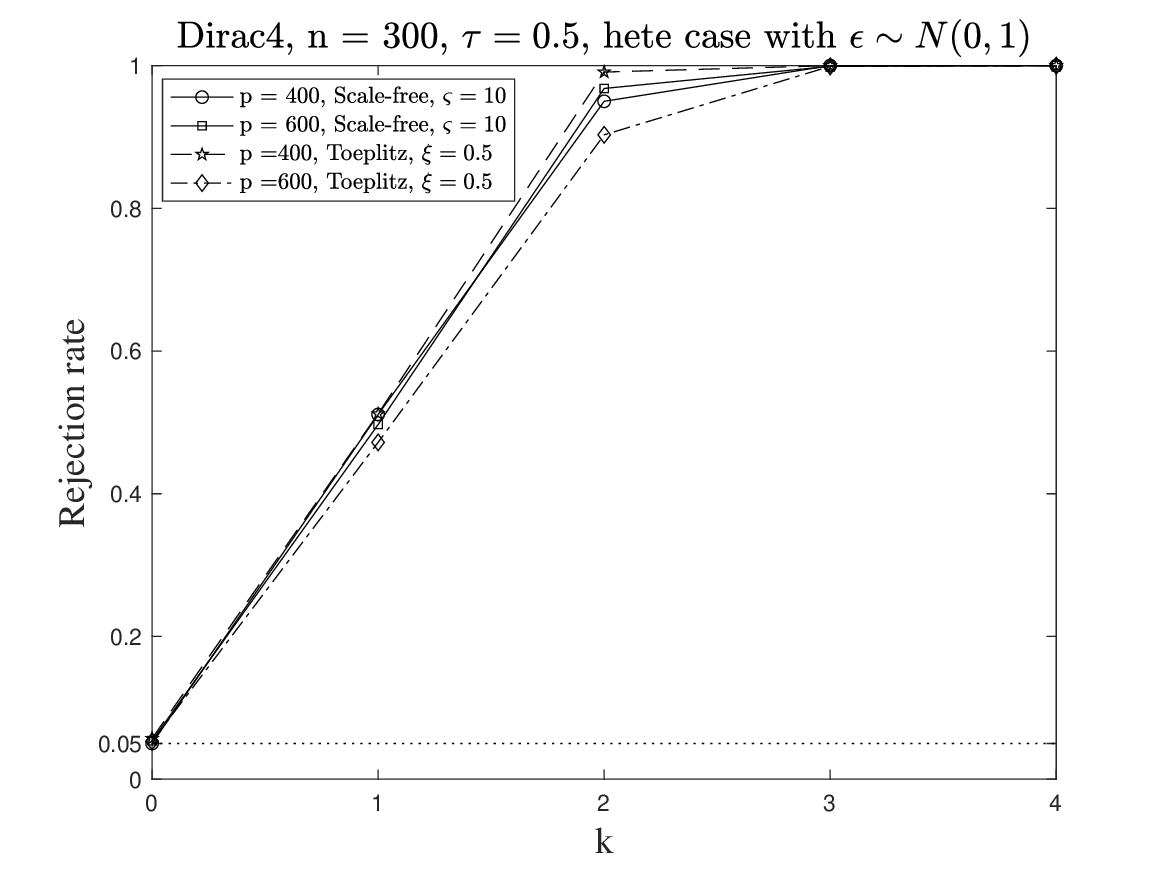}
	\end{minipage}
  \caption{The rejection rate of the proposed de-biased test for the expectile regression with different $\tau $ under the heteroscedasticity case with $n = 300 $ and $p=400,600$, $\beta^*_2 = k/\sqrt{n}$, and standard normal error, calculated from $1000$ replicates. The null hypothesis is ${\sf H}_0: \beta^*_2=0$.}
  \label{pic_hete_01}
\end{figure}

Figure \ref{pic_hete_01} depicts the empirical Type I error associated with the coefficient of $x_{(2)}$ across various expectile levels $\tau$, as $p$ increases from 400 to 600. It is evident that our proposed test effectively controls the empirical Type I error at both $\tau = 0.1$ and $\tau = 0.5$, across different dimensions $p$ and design setups, demonstrating the test's viability in heteroscedastic scenarios. Furthermore, it is observed that the empirical power reaches 1 when $\beta^*_2 = 4/\sqrt{n}$ for all scenarios depicted in Figure \ref{pic_hete_01}. This indicates that our proposed test statistic possesses considerable ability to detect potential active covariates in cases of heteroscedasticity.  Notably, in comparison with the homoscedastic case, the trend in local power becomes more pronounced, suggesting that heteroscedasticity significantly impacts the inference for some coordinate of the covariate closely associated with $x_{(1)} $.


Lastly, we consider group tests  \eqref{sig_gptest_Homo} with ${\cal G}_1 = \{1,2,3\}$ and ${\cal G}_2 = \{2,3,4\}$, where group ${\cal G}_1$ includes the critical variable 
$x_{(1)} $, while group ${\cal G}_2$ encompasses three variables predominantly affected by heteroscedasticity. Moreover, we employ the Dirac 4 pattern to generate the true coefficients.  In contrast to experimental designs with homoscedastic group tests, here, due to the nonlinear heteroscedastic structure of $x_{(1)} $, we only consider altering different dimensions $p$, expectile levels $\tau$ and types of errors to observe their rejection rates in group tests during simulations. In  Table \ref{hete_gp}, the rejection rates for both group tests, although lower than those of the previously presented univariate empirical Type I errors, are comparatively close. This demonstrates that our proposed multivariate testing approach maintains commendable performance under the heteroscedastic settings.


\begin{table}[htbp]
\small
  \centering
  \caption{The rejection rate calculated by different method, dimension $p$, expectile level $\tau$ and error type  for group ${\cal G}_1$ and ${\cal G}_2$ under the heteroscedastic case with $n = 300$,  $\xi = 0.50$ and  Toeplitz design  from $1000 $ replicates.}
    \begin{tabular}{cccccccccc}
    \toprule
    \toprule
          &       &       & \multicolumn{2}{c}{LA-LA } &       &       & \multicolumn{2}{c}{SC-SC } &  \\
\cmidrule{4-5}\cmidrule{8-9}          &       & \multicolumn{2}{c}{$\epsilon \sim \mathcal{N}(0,1)$} & \multicolumn{2}{c}{$\epsilon \sim t_4$} & \multicolumn{2}{c}{$\epsilon \sim \mathcal{N}(0,1)$} & \multicolumn{2}{c}{$\epsilon \sim t_4$} \\
\cmidrule{3-10}    group & $\tau$ value & $p = 400$ & $p = 600$ & $p = 400$ & $p = 600$ & $p = 400$ & $p = 600$ & $p = 400$ & $p = 600$ \\
    \midrule
          & $\tau = 0.1$ & 75.80\% & 73.40\% & 65.60\% & 56.40\% & 95.50\% & 94.10\% & 87.70\% & 83.50\% \\
    ${\cal G}_1= \{1,2,3\} $ & $\tau = 0.5$ & 4.80\% & 5.50\% & 4.60\% & 4.20\% & 4.60\% & 4.50\% & 4.40\% & 5.00\% \\
          & $\tau = 0.9$ & 78.90\% & 72.90\% & 63.70\% & 61.30\% & 92.60\% & 87.70\% & 86.40\% & 84.50\% \\
    \midrule
          & $\tau = 0.1$ & 4.60\% & 4.20\% & 4.30\% & 4.50\% & 4.50\% & 4.40\% & 3.90\% & 5.70\% \\
    ${\cal G}_2= \{2,3,4\}$  & $\tau = 0.5$ & 5.40\% & 5.40\% & 4.80\% & 4.70\% & 4.10\% & 4.90\% & 4.70\% & 4.30\% \\
          & $\tau = 0.9$ & 5.20\% & 6.20\% & 4.70\% & 6.10\% & 4.50\% & 4.70\% & 4.90\% & 4.50\% \\
    \bottomrule
    \bottomrule
    \end{tabular}%
  \label{hete_gp}%
\end{table}%

 \section{Real data analysis}\label{sec6}
\subsection{Financial and macro data}

In this empirical study, we examine the relationship between monetary policy measures and stock market prices, using the FRED-MD dataset, available on the Fred-MD website
\footnote{https://research.stlouisfed.org/econ/mccracken/fred-databases/}.
The dataset consists of 127 U.S. macroeconomic variables observed monthly over the period from January 1959 to September 2023. These macroeconoic variables can be classified into eight groups: consumption, orders and inventories; housing; interest and exchange rates; labour market; money and credit; output and income; prices; and the stock market. More detailed description can be found in \cite{MN16}.

We set the return of S\&P 500 as the dependent variable and the other 126 variables with lags 1, 2, and 3 as independent variables. 
Table \ref{table_md_empirical} reports the results of the de-biased estimation and inference under the expectile levels $\tau = 0.1,0.3,0.5,0.7,$ and $0.9$. 
The estimates of the regression coefficients related to  Monetary Base (BOGMBASE) and M1 Money Stock (M1SL) are reported with standard deviation. Furthermore, it presents the $p$-values for separate group tests on BOGMBASE and M1SL with lags 1, 2, and 3 using our proposed test statistics.  Table \ref{table_md_empirical} shows that the de-biased estimates of M1SL${}_{t-1}$ and M1SL${}_{t-3}$ are statistically significant at 10$\%$ significance level under the $\tau = 0.9$ case while the de-biased estimates of the regression coefficients  related to M1SL are not significant under other expectile levels. Moreover, the group test on the M1SL is statistically significant at 5$\%$ significance level under the $\tau = 0.9$ case, which indicates that the M1SL can be used to predict the upper expectiles of the distribution of S\&P 500 returns. Furthermore,  the differences in the results of the de-biased estimates at given expectile levels, along with the disparity in the significance of group tests for variables BOGMBASE and M1SL, suggest that the regression model exhibits a certain degree of  heteroscedasticity. This finding  is aligned with the  results reported by \cite{Taa15}, who showed that money supply has no impact on stock prices using a nonparametric Granger causality in mean test, but when employing quantile regression-based test, the impact of money supply on stock prices becomes apparent and statistically significant. { They showed that using a nonparametric Granger causality test in mean, money supply appears to have no impact on stock prices. However, when a quantile regression-based test is employed, the influence of money supply on stock prices becomes apparent and statistically significant.}



\begin{table}[htbp]
\small
  \centering
  \caption{High-dimensional expectile regression results of S\&P 500 on the other 126 variables with lags 1--3 as independent variables. $^{***}$", "$^{**}$", "$^*$" indicate the statistical significance at 1$\%$, 5$\%$, and 10$\%$ levels, respectively.}
    \begin{tabular}{lcccccc}
    \toprule
    \toprule
          & \multicolumn{1}{c}{$\tau = 0.1$} & \multicolumn{1}{c}{$\tau = 0.3$} & \multicolumn{1}{c}{$\tau = 0.5$} & \multicolumn{1}{c}{$\tau = 0.7$} & \multicolumn{1}{c}{$\tau = 0.9$}& \\\cmidrule{2-6}   
M1SL${}_{t-1}$ & -0.021  & -0.023  & -0.032  & -0.049  & -0.072${}^*$  \\
          & (0.1065) & (0.0742) & (0.0599) & (0.0511) & (0.0423) \\
    M1SL${}_{t-2}$ & 0.105  & 0.066  & 0.056  & 0.030  & 0.002  \\
          & (0.1085) & (0.0709) & (0.0557) & (0.0468) & (0.0399) \\
    M1SL${}_{t-3}$ & -0.135  & -0.070  & -0.065  & -0.074  & -0.084${}^*$  \\
          & (0.0870) & (0.0538) & (0.0504) & (0.0484) & (0.0459) \\
    BOGMBASE${}_{t-1}$ & -0.022  & 0.018  & 0.036  & 0.043  & 0.048  \\
          & (0.1288) & (0.0919) & (0.0768) & (0.0638) & (0.0568) \\
    BOGMBASE${}_{t-2}$ & 0.058  & 0.028  & 0.013  & 0.008  & 0.019  \\
          & (0.0742) & (0.0600) & (0.0540) & (0.0482) & (0.0417) \\
    BOGMBASE${}_{t-3}$ & -0.090  & -0.040  & -0.023  & -0.010  & -0.021  \\
          & (0.0733) & (0.0510) & (0.0436) & (0.0430) & (0.0389) \\
    \midrule
    \multicolumn{7}{l}{H${}_0$: $\beta_{{\textit M1SL}{}_{t-1}}=\beta_{\textit{M1SL}{}_{t-2}}=\beta_{\textit{M1SL}{}_{t-3}}=0$} \\
    p-value & 0.444  & 0.479  & 0.350  & 0.188  & 0.047  \\
    \midrule
    \multicolumn{7}{l}{H${}_0$: $\beta_{{\textit BOGMBASE}{}_{t-1}}=\beta_{\textit{BOGMBASE}{}_{t-2}}=\beta_{\textit{BOGMBASE}{}_{t-3}}=0$} \\
    p-value & 0.625  & 0.837  & 0.889  & 0.897  & 0.752  \\
    \bottomrule
    \bottomrule
    \end{tabular}%
    \label{table_md_empirical}%
\end{table}%

 \subsection{Gene data}
 
 In this section, we use a microarray dataset to demonstrate the effectiveness of our proposed de-biased expectile test statistic. \cite{HLP11} investigated the impact of the innate immune system on the development of atherosclerosis by analyzing gene profiles from 119 patients' peripheral blood. The data, collected using the Illumina HumanRef8 V2.0 Bead Chip, are publicly available on the Gene Expression Omnibus. Following \cite{HLP11}, the toll-like receptors (TLR) signaling pathway is identified as crucial in activating the innate immune response in atherosclerosis cases. Specifically, the 'TLR8' gene is highlighted as a key gene associated with atherosclerosis. \cite{ZCZ18} further demonstrate that the skewness and kurtosis of the 'TLR8' gene data deviate from the normal distribution. Notably, they applied the Lilliefors Test to the 'TLR8' gene data and concluded that it does not conform to a normal distribution at the 5$\%$ significance level. To further study the relation between this key gene and the other genes, \cite{FLW17} regressed it on another 464 genes from 12 different related-pathways (TLR, CCC, CIR, IFNG, MAPK, RAPO, EXAPO, INAPO, DRS, NOD, EPO, CTR) and their result reveals that the regression residuals have heavy right tail and skewed distribution. Motivated by such heteroscedastic feature of the data, we  apply our proposed de-biased expectile test on each possible gene at  different expectile levels, $\tau =0.1, 0.3, 0.5, 0.7$, and $0.9 $. 

\begin{table}[htbp]
\small
  \centering
  \caption{Significant genes in the expectile regression at the significance level $\alpha = 0.05 $, obtained by the de-biased method with {$\tau \in \{ 0.1, 0.3, 0.5, 0.7, 0.9 \}$}. The genes found in \cite{FLW17} are in boldface.  }
    \begin{tabular}{cccccccccc}
    \toprule
    \toprule
    \multicolumn{2}{c}{$\tau =0.1$} & \multicolumn{2}{c}{$\tau = 0.3$} & \multicolumn{2}{c}{$\tau = 0.5$} & \multicolumn{2}{c}{$\tau = 0.7$} & \multicolumn{2}{c}{$\tau = 0.9$} \\
    \midrule
    \textbf{IFI6} & 0.142  & \textbf{IFI6} & 0.148  & \textbf{MAPK1} & 0.299  & IFNA2 & -0.162  & IL1A  & -0.314  \\
    RASGRP1 & 0.212  & \textbf{MAPK1} & 0.265  & MAPK14 & 0.241  & MAPK10 & -0.101  & \textbf{CSF3} & -0.227  \\
    \textbf{MAPK1} & 0.171  & MAPK14 & 0.211  & \textbf{PSMB8} & 0.141  & MAPK7 & -0.180  & IFNGR2 & -0.206  \\
    PRKCH & -0.164  & \textbf{PSMB8} & 0.146  & \textbf{KPNB1} & 0.198  & \textbf{MAPK1} & 0.209  & IFNA2 & -0.180  \\
    DUSP1 & 0.136  & \textbf{KPNB1} & 0.178  & \textbf{BCL2L11} & -0.184  & PSMA5 & -0.198  & \textbf{AKT3} & -0.172  \\
    \textbf{AKT1} & 0.231  & CASP6 & -0.135  & CFLAR & 0.365  & SPTAN1 & -0.232  & MAP3K14 & -0.393  \\
    \textbf{TLR3} & 0.121  & \textbf{BCL2L11} & -0.208  & \textbf{CRK } & 0.250  & \textbf{KPNB1} & 0.182  & MAPK7 & -0.202  \\
    \textbf{DAPK2} & 0.129  & CFLAR & 0.277  &       &       & CFLAR & 0.421  & NRAS  & -0.345  \\
    PSMD4 & -0.176  & \textbf{CRK } & 0.203  &       &       & CBL   & -0.174  & MAP3K11 & -0.133  \\
    \textbf{PSMB8} & 0.152  &       &       &       &       & PTK2B & -0.097  & MAP3K4 & -0.264  \\
    STK24 & 0.231  &       &       &       &       & \textbf{CRK } & 0.370  & AKT2  & -0.199  \\
    PRKCQ & -0.130  &       &       &       &       &       &       & PSMD6 & -0.244  \\
    \textbf{KPNB1} & 0.216  &       &       &       &       &       &       & UBA52 & -0.181  \\
    HIST1H1D & -0.165  &       &       &       &       &       &       & PSMC1 & -0.149  \\
    CASP6 & -0.186  &       &       &       &       &       &       & VIM   & -0.178  \\
    \textbf{BCL2L11} & -0.238  &       &       &       &       &       &       & CFLAR & 0.576  \\
    AIM2  & 0.202  &       &       &       &       &       &       & NOD2  & -0.230  \\
    ZAP70 & -0.113  &       &       &       &       &       &       & CBL   & -0.220  \\
          &       &       &       &       &       &       &       & IL17A & -0.188  \\
          &       &       &       &       &       &       &       & \textbf{CRK } & 0.535  \\
          &       &       &       &       &       &       &       & CDC42 & -0.188  \\
          &       &       &       &       &       &       &       & CD247 & -0.239  \\
    \bottomrule
    \bottomrule
    \end{tabular}%
  \label{gene_table1}%
\end{table}%

Table \ref{gene_table1} reports the significant genes along with their de-biased estimators. For clarity, we bold the overlapping genes found in \cite{FLW17}. In the $\tau = 0.5 $ case, where the expectile regression is equivalent to the mean regression, there are totally seven genes significant at $\alpha = 5\% $.  Our approach identifies six additional genes compared with the sole gene CRK obtained using the Lasso estimator in \cite{FLW17}.   Furthermore,  the significant genes selected by our proposed de-biased method at different expectile levels include a subset of genes identified by the adaptively weighted lasso (R-Lasso) method in \cite{FYB14} and the regularized approximate quadratic estimator with an $L_1$ regularizer (RA-Lasso) in \cite{FLW17}.  Notably, at $\tau = 0.1$ and $\tau = 0.9$, the significant genes identified overlap with those selected at $\tau = 0.5$, aligning with the findings in \cite{ZCZ18}.  Moreover, both the significant genes and their corresponding de-biased estimators vary with changing expectile levels, indicating the presence of heteroscedasticity in this data.

In summary, examining the heteroscedastic pattern through expectile linear regression at various expectile levels, along with the identification of significant covariates by our proposed test, may provide new insights into depicting the relationship between the target gene 'TLR8' and other related genes. 

%




\section{Conclusion}\label{sec7}
In this article, we develop a bias correction procedure to conduct statistical inference for high-dimensional  sparse linear models under the expectile settings. Due to the second-order non-Lipschitz property of the expectile loss function,  our approach deviates from that of \cite{VBRD14}. An alternative proving strategy is developed to demonstrate that the estimation errors in the preliminary estimator have little impact on the estimation of the expectile-specified random weights, which is pivotal in revealing the asymptotic normality of the  de-biased estimator. A Wald-type test statistic is then established for multivariate testing.  Simulation results indicate that the test we proposed performs well in controlling the Type I error rate and demonstrates good local power under both the homoscedastic and heteroscedastic cases, even with heavy-tailed errors. Furthermore, compared with the method studied in \cite{ZCZ18}, The simulation demonstrates that our proposed test statistic possesses ability to detect heteroscedasticity based on different expectile levels. 

The empirical study on the FRED-MD dataset demonstrates that under the expectile settings, monetary supply may possess a certain predictive ability for stock prices, which is aligned with the conclusions in \cite{Taa15}.  Meanwhile, the empirical result on the selected gene data reveals a list of genes that exhibit significant explanatory ability for the 'TLR8' gene, which, along with the corresponding de-biased estimators vary across different expectile levels. Such findings not only highlight the presence of heteroscedasticity within the data, which has been discussed in \cite{ZCZ18}, but also provide a novel perspective on the relationships between the 'TLR8' gene and other genes which consists a subset of genes identified in \cite{FLW17}.

It is important to note that our results are applicable to a broader regularization framework, where the regularizers can be either convex or non-convex, provided they belong to the amenable category. While the use of non-convex regularizers introduces complexity to the optimization problem, their superior performance in variable selection may allow for more effective control over Type I errors and yields higher local power compared with convex regularizers. This assertion is further corroborated by our simulation results, particularly in scenarios with less sparse Hessian matrices. However, when addressing hypothesis testing in general sparse scenarios, even though the results using convex and non-convex regularizers are similar, it is crucial to note the considerably heavy  computational burden brought about by the use of non-convex regularizers.

Future research directions involve extending our theoretical results to facilitate simultaneous inference for high-dimensional components of a large parameter vector in sparse expectile linear models, in which the dimension of the parameter vector of interest is allowed to grow with the sample size.
Moreover, we acknowledge that our assumption regarding the existence of the 8th moment of the observations can be overly restrictive for certain datasets. Therefore, there is a need for the development of more robust expectile regression methods and Hessian matrix estimation methods to address datasets with less stringent moment assumptions.



\section*{Acknowledgments}

We thank the Editor, Associate Editor and referees. Zhang's research was supported by grants from the NSF of China (Grant Nos.U23A2064 and 12031005).  Zhao's research was supported by the MOE Project of Humanities and Social Sciences (No. 21YJCZH235), the Hangzhou Joint Fund of the Zhejiang Provincial Natural Science Foundation of China under Grant No. LHZY24A010002 and Scientific Research Foundation of Hangzhou City University (No.J-202315).

\section*{Appendix A: Proof of the main results}

\renewcommand{\thesubsection}{A.\arabic{subsection}}
\renewcommand{\theequation}{A.\arabic{equation}} 
\renewcommand{\thelemma}{A.\arabic{lemma}}
\setcounter{equation}{0}
\setcounter{lemma}{0}

We start with two auxiliary lemmas, where the proof can be found in \cite{LW15}.

\begin{lemma}
\label{lux_amen_convex}
  If $P_{\lambda}: \mathbb{R}^{p} \to \mathbb{R}$ is an amenable regularizer defined in Definition \ref{amenable-p}, then  the function $  \lambda \Vert \boldsymbol{\beta} \Vert_1 - P_{\lambda}(\boldsymbol{\beta}) $ with respect to $\boldsymbol{\beta}  \in \mathbb{R}^{p} $ is everywhere differentiable. Moreover, the function $\mu \Vert \boldsymbol{\beta}  \Vert_{2}^{2}/2 + P_{\lambda}(\boldsymbol{\beta}) -  \lambda \Vert \boldsymbol{\beta} \Vert_1 $ is convex and for any $\boldsymbol{\beta}  \in \mathbb{R}^{p} $, we have 
 $$ \lambda \Vert \boldsymbol{\beta}  \Vert_{1} \leq P_{\lambda}(\boldsymbol{\beta} ) +  \frac{\mu}{2}\Vert \boldsymbol{\beta}  \Vert_{2}^{2}. $$
\end{lemma}

\begin{lemma}
\label{lux_amen_L1}
  Suppose that the regularizer $P_{\lambda} $ satisfies the conditions (i)-(vi) in  Definition \ref{amenable-p}.  If $\boldsymbol{\beta}^* \in \mathbb{R}^p $ is $k-$sparse, then for any $\boldsymbol{\beta} \in \mathbb{R}^p $ such that $\xi P_{\lambda}(\boldsymbol{\beta}^*) - P_{\lambda}(\boldsymbol{\beta})>0 $   and $\xi \geq 1$, we have
  $$
  \xi P_{\lambda}(\boldsymbol{\beta}^*) - P_{\lambda}(\boldsymbol{\beta}) \leq \xi \Vert \Delta\boldsymbol{\beta}_{\cal B} \Vert_1 - \Vert \Delta\boldsymbol{\beta}_{{\cal B}^C}\Vert_1,
  $$
  where $\Delta\boldsymbol{\beta} :=\boldsymbol{\beta}^* - \boldsymbol{\beta} $ and ${\cal B}$ is the index set of the k largest elements of $\boldsymbol{\beta}$ in magnitude. 
\end{lemma}

\subsection{Proof of Proposition \ref{prop1}}
\begin{lemma}\label{lemmaA.2}
Under Assumption \ref{assump1}, we have 
\begin{equation}\label{eqAA.2}
\Vert \hat{\boldsymbol{\Sigma}} - \boldsymbol{\Sigma}\Vert_{\infty}={ O_p(  \sqrt{\ln p/n}).} 
\end{equation}
\end{lemma}
\begin{proof}
Define $\sigma_{i,jk}=x_{ij}x_{ik}$, $\bar\sigma_{i,jk}=\sigma_{i,jk}I\{|\sigma_{i,jk}|\leq CK^2\}$ and $\tilde\sigma_{i,jk}=\sigma_{i,jk}I\{|\sigma_{i,jk}|>CK^2\}$.
 Then by Bernstein inequality and union bound inequality, we have 
 $$
\Pr\left(  \max_{1\leq j,k\leq p}\left|\frac{1}{n} \sum_{i}^n \left(\bar\sigma_{i,jk} - {\sf E} [\bar\sigma_{i,jk}]\right)  \right| \geq t \right) \leq p^2\exp\left[  -\frac{(nt)^2/2}{\kappa^2 + CK^2nt/3}  \right],$$
where $\kappa^2 = \sum^n_{i = 1} {\sf Var}[\bar\sigma_{i,jk}] = O(n)$.
And this further implies that 
\begin{equation}\label{eqAA.3}
\max_{1\leq j,k\leq p}\left|\frac{1}{n} \sum_{i}^n \left(\bar\sigma_{i,jk} - {\sf E} [\bar\sigma_{i,jk}]\right)  \right| = O_p(  \sqrt{\ln p/n} ) +  O_p(K^{2}  {\ln p/n} ).
\end{equation}
On the other hand,
\begin{equation*}
 |{\sf E} 
[\tilde\sigma_{i,jk}]|
\leq ({\sf E} [\sigma_{i,jk}^2] \Pr \{|\sigma_{i,jk}|>CK^2\})^{1/2}=o(\sqrt{\ln p/n}),   
\end{equation*}
then we have 
 \begin{eqnarray}
&&\Pr\left(  \max_{1\leq j,k\leq p}\left|\frac{1}{n} \sum_{i}^n \left(\tilde\sigma_{i,jk} - {\sf E} [\tilde\sigma_{i,jk}]\right)  \right| > C \max \{ \sqrt{\ln p/n}, K^2\ln p/n\}  \right) \nonumber\\
&\leq&\Pr\left(  \max_{1\leq j,k\leq p}\left|\frac{1}{n} \sum_{i}^n \tilde\sigma_{i,jk}  \right| > \frac{C}{2} \max \{ \sqrt{\ln p/n}, K^2\ln p/n\} \right) \nonumber\\
&{\leq }&\Pr\left(  \max_{1\leq j,k\leq p} \max_{1\leq i\leq n}\left|\sigma_{i,jk}  \right| > C K^2 \right)=o(1), \label{eqAA.4}
 \end{eqnarray}
where the last step follows from the fact that
$$\Pr\left(  \max_{1\leq j,k\leq p} \max_{1\leq i\leq n}\left|\sigma_{i,jk}  \right| > C K^2 \right) = \Pr\left(  \max_{1\leq j\leq p} \max_{1\leq i\leq n}\left|x_{i,j}  \right| > \sqrt{C} K \right)$$
and the assumption $\Vert \mathbf{X}\Vert_{\infty} = O_{p}(K)$.

Combing \eqref{eqAA.3}, \eqref{eqAA.4}, and the assumption $ K^2 \sqrt{\ln p/n} = o(1)$, we prove \eqref{eqAA.2}.
\end{proof}

\medskip
\begin{lemma}\label{lemmaA.3}
Under Assumptions \ref{assump1}, we have 
\begin{equation}
\Vert\mathbf{X}^\top\mathbf{W}_{\boldsymbol{\beta} ^*}^2\boldsymbol{\epsilon}/n\Vert_{\infty}={  O_p(\sqrt{\ln p/n})}.
\end{equation}
\end{lemma}
\begin{proof}
Following a similar argument as in the proof of Lemma \ref{lemmaA.2}, we can prove the result. 
\end{proof}



\begin{proof}[\textbf{\upshape Proof of Proposition \ref{prop1}}]
Recall that $\nabla L_n(\boldsymbol{\beta} ^*)=- \mathbf{X}^\top\mathbf{W}_{\boldsymbol{\beta} ^*}^2\boldsymbol{\epsilon}/n $ and let $z^*_{\infty} = \Vert \nabla L_n(\boldsymbol{\beta} ^*)\Vert_{\infty} $.
By Lemma 4 of \cite{GZ16} and letting $\Delta\boldsymbol{\beta}  = \hat{\boldsymbol{\beta}}   - \boldsymbol{\beta} ^* $, we can prove that 
\begin{equation}\label{eqA.1}
\min\{\tau,1-\tau\}\Vert \mathbf{X} \Delta\boldsymbol{\beta}  \Vert^2_2 /n \leq  \langle \nabla L_n(\hat{\boldsymbol{\beta}}  )  - \nabla L_n(\boldsymbol{\beta} ^*),  \Delta\boldsymbol{\beta}  \rangle.
\end{equation}
On the left hand side, we have
\begin{equation}
\Vert \mathbf{X} \Delta\boldsymbol{\beta}  \Vert^2_2 /n = \Delta\boldsymbol{\beta}^\top \Sigma\Delta\boldsymbol{\beta} +  \Delta\boldsymbol{\beta}^\top (\hat{\Sigma} - \Sigma)\Delta\boldsymbol{\beta}\geq \lambda_{\min}(\boldsymbol{\Sigma})\Vert \Delta \boldsymbol{\beta} \Vert^2_2  - \Vert \boldsymbol{\Sigma}  -   \hat{\boldsymbol{\Sigma}} \Vert_{\infty} \Vert \Delta \boldsymbol{\beta} \Vert^2_1 ,\label{A1.01}
\end{equation}
and on the right hand side, we have
\begin{equation}\label{A3.1}
\langle \nabla L_n(\hat{\boldsymbol{\beta}}  )  - \nabla L_n(\boldsymbol{\beta} ^*),  \Delta\boldsymbol{\beta}  \rangle = \langle -\nabla P_{\lambda}(\hat{\boldsymbol{\beta}}  )  - \nabla L_n(\boldsymbol{\beta} ^*),  \Delta\boldsymbol{\beta}  \rangle \leq  z^*_{\infty}\Vert \Delta \boldsymbol{\beta} \Vert_1 + P_{\lambda}(\boldsymbol{\beta}^*) - P_{\lambda}(\hat{\boldsymbol{\beta}}) + \frac{\mu}{2}\Vert \Delta\boldsymbol{\beta} \Vert^2_2, 
\end{equation}
where the last inequality follows from the property of $\mu-$amenable regularizer, i.e., $$\langle \nabla P_{\lambda}(\hat{\boldsymbol{\beta}}) , -\Delta\boldsymbol{\beta} \rangle \leq P_{\lambda}(\boldsymbol{\beta}^*) - P_{\lambda}(\hat{\boldsymbol{\beta}}) + \frac{\mu}{2}\Vert \Delta\boldsymbol{\beta} \Vert^2_2.$$
Combining \eqref{eqA.1}--\eqref{A3.1}, we obtain 
\begin{equation*}\label{A.51}
\begin{aligned}
\min\{\tau,1-\tau\}\lambda_{\min}(\boldsymbol{\Sigma})\Vert \Delta \boldsymbol{\beta} \Vert^2_2  & \leq (z^*_{\infty} + \min\{\tau,1-\tau\}\Vert \boldsymbol{\Sigma}  -   \hat{\boldsymbol{\Sigma}} \Vert_{\infty} \Vert \Delta \boldsymbol{\beta} \Vert_1 )\Vert \Delta \boldsymbol{\beta} \Vert_1 + P_{\lambda}(\boldsymbol{\beta}^*) - P_{\lambda}(\hat{\boldsymbol{\beta}}) + \frac{\mu}{2}\Vert \Delta\boldsymbol{\beta} \Vert^2_2 \\
& \leq (z^*_{\infty} + 2\Vert \boldsymbol{\Sigma}  -   \hat{\boldsymbol{\Sigma}} \Vert_{\infty} R )\cdot\frac{1}{\lambda} (P_{\lambda}(\Delta \boldsymbol{\beta}) + \frac{\mu}{2}\Vert \Delta \boldsymbol{\beta} \Vert^2_2  )  + P_{\lambda}(\boldsymbol{\beta}^*) - P_{\lambda}(\hat{\boldsymbol{\beta}}) + \frac{\mu}{2}\Vert \Delta\boldsymbol{\beta} \Vert^2_2,
\end{aligned}
\end{equation*}
where the last step follows from $\Vert \Delta\boldsymbol{\beta} \Vert_1 \leq \Vert \hat{\boldsymbol{\beta}} \Vert_1 +\Vert {\boldsymbol{\beta}^*} \Vert_1 \leq  2R$ and Lemma \ref{lux_amen_convex}. 
Note further that the regularizer is sub-addictivie, i.e., 
$
P_{\lambda}(\Delta\boldsymbol{\beta}) \leq P_{\lambda}(\boldsymbol{\beta}^*) + P_{\lambda}(\hat{\boldsymbol{\beta}})
$
and 
$\min\{\tau,1-\tau\}\lambda_{\min}(\boldsymbol{\Sigma}) > 3\mu/4 $. 
Under the event   
$$ {\cal E}=
 \{\lambda\geq 4z^*_{\infty}\} \cap\{ \lambda\geq 8\Vert \boldsymbol{\Sigma}  -   \hat{\boldsymbol{\Sigma}} \Vert_{\infty} R\},
$$
 we deduce that
\begin{equation}\label{eq.A.10.1}
0 \leq  (\min\{\tau,1-\tau\}\lambda_{\min}(\boldsymbol{\Sigma}) - \frac{3}{4}\mu) \Vert \Delta \boldsymbol{\beta} \Vert^2_2  \leq \frac{3}{2}P_{\lambda}(\boldsymbol{\beta}^*) - \frac{1}{2}P_{\lambda}(\hat{\boldsymbol{\beta}}). 
\end{equation}
Denote ${\cal A}_1$ the index set of the $s$ largest components of $\Delta\boldsymbol{\beta}$ in magnitude, which may be slightly different from the active set ${\cal A}$ in most cases.  By the property of the $\mu-$amenable regularizer
(see Lemma 5 in \cite{LW15}), we have
\begin{equation}\label{eq.A.10.12}
0< \frac{3}{2}P_{\lambda}(\boldsymbol{\beta}^*) - \frac{1}{2}P_{\lambda}(\hat{\boldsymbol{\beta}}) \leq \frac{3}{2} \lambda \Vert \Delta\boldsymbol{\beta}_{{\cal A}_1} \Vert_1 - \frac{1}{2} \lambda \Vert \Delta\boldsymbol{\beta}_{{\cal A}^C_1} \Vert_1 \leq \frac{3}{2}\sqrt{s} \lambda \Vert \Delta\boldsymbol{\beta}_{{\cal A}_1} \Vert_2 \leq  \frac{3}{2} \sqrt{s} \lambda  \Vert \Delta\boldsymbol{\beta}  \Vert_2.
\end{equation}
Then by combining \eqref{eq.A.10.1} and \eqref{eq.A.10.12},  it yields the $l_2-$bound 
$$
\Vert \Delta\boldsymbol{\beta} \Vert_2 \leq \frac{6\sqrt{s}\lambda }{4\min\{\tau,1-\tau\}\lambda_{\min}(\boldsymbol{\Sigma}) - 3\mu },
$$
and the $l_1-$bound can be derived using \eqref{eq.A.10.12}, i.e.,
$$
\Vert \Delta\boldsymbol{\beta} \Vert_1 = \Vert \Delta\boldsymbol{\beta}_{{\cal A}_1} \Vert_1 + \Vert \Delta\boldsymbol{\beta}_{{\cal A}^C_1} \Vert_1 \leq 4\Vert \Delta\boldsymbol{\beta}_{{\cal A}_1} \Vert_1 \leq 4\sqrt{s}\Vert \Delta\boldsymbol{\beta} \Vert_2. $$
Next, from \eqref{eqA.1} and \eqref{A3.1} we can obtain the following result for the prediction error bound,
$$
\Vert \mathbf{X} \Delta\boldsymbol{\beta}  \Vert^2_2 /n \leq \frac{s\lambda^2}{\min\{\tau,1-\tau\}} \left ( \frac{12}{4\min\{\tau,1-\tau\}\lambda_{\min}(\boldsymbol{\Sigma}) - 3\mu} + \frac{36\mu}{(4\min\{\tau,1-\tau\}\lambda_{\min}(\boldsymbol{\Sigma}) - 3\mu)^2}\right).
$$

Finally, by Lemmas \ref{lemmaA.2} and \ref{lemmaA.3}
and  choosing $\lambda$  such that $\lambda\geq c_7R\sqrt{\ln p/n}$ for some large positive constant $c_7$, we can prove that $\Pr ({\cal E})\to 1$, and  we  complete the proof of Proposition \ref{prop1}. 

\end{proof}


\bigskip
\begin{proof}[\textbf{\upshape Proof of Lemma \ref{lem1}.}]
Since 
\begin{eqnarray}
|w^2_{\boldsymbol{\beta} ^*,i} - w^2_{\hat{\boldsymbol{\beta}}  ,i}| &=& |1 - 2\tau|\cdot \mathbb{I}( \epsilon_i \hat{\epsilon}_i <0)\nonumber\\
&\leq& \mathbb{I}(|\epsilon_i| \leq |\hat{\epsilon}_i - \epsilon_i|)
= \mathbb{I}(|\epsilon_i| \leq |\boldsymbol{X}_i^\top (\hat{\boldsymbol{\beta}}   - \boldsymbol{\beta} ^*) |),\nonumber
\end{eqnarray}
holds for any $i\in\{1,\ldots, n\} $. Then by the Markov's inequality,
we have  
\begin{equation}\label{eqA.10new}
\Pr  \left( \frac{1}{n}\sum_{i=1}^n |w^2_{\hat{\boldsymbol{\beta}}  ,i} - w^2_{\boldsymbol{\beta} ^*,i}|  \geq \delta_1 \right)
\leq\frac{1}{n\delta_1}\sum_{i=1}^n\Pr \left(\left|\boldsymbol{X}_{i}^\top  (\hat{\boldsymbol{\beta}} -\boldsymbol{\beta} ^*)\right|>\delta_2 \right) 
+\frac{1}{n\delta_1}\sum_{i=1}^n\Pr \left(|\epsilon_i|\leq\delta_2 \right),
\end{equation}
for any  $\delta_1, \delta_2>0 $. 
Note that by the Holder's inequality, 
\begin{equation}\label{eqA.11new}
\frac{1}{n}\sum_{i=1}^n\Pr \left(\left|\boldsymbol{X}_{i}^\top  (\hat{\boldsymbol{\beta}} -\boldsymbol{\beta} ^*)\right|>\delta_2 \right) 
\leq\frac{1}{n}\sum_{i=1}^n\Pr \left(\Vert \boldsymbol{\hat{\beta}}   - \boldsymbol{\beta} ^* \Vert_{1}\Vert \boldsymbol{X}_i \Vert_{\infty}>\delta_2 \right)=\Pr \left(\Vert \hat{\boldsymbol{\beta}}   - \boldsymbol{\beta} ^* \Vert_{1}\Vert \boldsymbol{X}_i \Vert_{\infty}>\delta_2 \right)
\end{equation} 
and
\begin{equation}\label{eqA.12new}
\frac{1}{n}\sum_{i=1}^n\Pr \left(|\epsilon_i|\leq\delta_2 \right) \leq  2\delta_2 \cdot \sup_{x \in(-\delta_2,+\delta_2)}f_\epsilon(x) \leq 2\delta_2 \cdot \sup_{x \in(-\infty,+\infty)}f_\epsilon(x).
\end{equation}
Thus, by combining \eqref{eqA.10new}--\eqref{eqA.12new}, and taking $\delta_1=c_9\delta_2 $ and $\delta_2=c_{10}Ks\lambda$, where $c_9$ and $c_{10}$ are large enough positive constants, 
we can prove \eqref{eqT42.4}.
\end{proof}


\bigskip

\subsection{Proof of Theorem \ref{thm1}}


\begin{proof}[\textbf{\upshape Proof of Theorem \ref{thm1}.}]
(i) We start with the $\mathbf{W}_{\boldsymbol{\beta} ^*}$-weighted node-wise regression, that is  
\begin{equation}\label{lemma4.1_01}
  \boldsymbol{X}_{\boldsymbol{\beta} ^*,(j)} = \mathbf{X}_{\boldsymbol{\beta} ^*,(-j)} \boldsymbol{\varphi}_{\boldsymbol{\beta} ^*,j} +\boldsymbol{\varrho}_{\boldsymbol{\beta} ^*,j} ,
\end{equation}
where $\boldsymbol{X}_{\boldsymbol{\beta} ^*,(j)} $ is the $j-$th column of $\mathbf{X}_{\boldsymbol{\beta} ^*} = \mathbf{W}_{\boldsymbol{\beta} ^*}\mathbf{X} $,  and $\boldsymbol{\varphi}_{\boldsymbol{\beta} ^*,j} $ is defined by
\begin{equation*}\label{eq.A11}
  \boldsymbol{\varphi}_{\boldsymbol{\beta} ^*,j} = \argmin_{\boldsymbol{\varphi}} {\sf E}\Vert \boldsymbol{X}_{\boldsymbol{\beta} ^*,(j)} -\mathbf{X}_{\boldsymbol{\beta} ^*,(-j)} \boldsymbol{\varphi} \Vert^2_{2} = \argmin_{\boldsymbol{\varphi}} {\sf E}[( x_{\boldsymbol{\beta} ^*,(j),1} -\mathbf{X}_{\boldsymbol{\beta} ^*,(-j),1}^\top \boldsymbol{\varphi} )^2],
\end{equation*}
where ${x}_{\boldsymbol{\beta} ^*,(j),1} $ is the first element of $\boldsymbol{X}_{\boldsymbol{\beta} ^*,(j)} $ and $\mathbf{X}_{\boldsymbol{\beta} ^*,(-j),1}^\top$ is the first row of $\mathbf{X}_{\boldsymbol{\beta} ^*,(-j)}$. 
Furthermore, we denote the variance of the residual by 
\begin{equation*}
  \phi^2_{\boldsymbol{\beta} ^*,j} = {\sf E} \Vert \boldsymbol{\varrho}_{\boldsymbol{\beta} ^*,j} \Vert^2_{2}/n ={\sf E}[\boldsymbol{\varrho}^2_{\boldsymbol{\beta} ^*,j,1}], \quad j\in\{1,\ldots, p\}.
\end{equation*}
Thus we have $  \phi^2_{\boldsymbol{\beta} ^*,j} =1/{\Theta}_{\boldsymbol{\beta} ^*,jj}$.
 Recall that $\boldsymbol{\Theta}_{\boldsymbol{\beta} ^*,j} = (-{\varphi}_{\boldsymbol{\beta} ^*,j,1}/\phi^2_{\boldsymbol{\beta} ^*,j},...,  -{\varphi}_{\boldsymbol{\beta} ^*,j,p}/\phi^2_{\boldsymbol{\beta} ^*,j}   )^\top $. Using the eigenvalue condition, $$ 1/\lambda_{\max}( {\boldsymbol{\Sigma}}_{\boldsymbol{\beta} ^*}) \leq {\Theta}_{\boldsymbol{\beta} ^*,jj} \leq 1/\lambda_{\min}( {{\boldsymbol{\Sigma}}_{\boldsymbol{\beta} ^*}}), $$ and  
 $$ \lambda_{\min}({\boldsymbol{\Sigma}}_{\boldsymbol{\beta} ^*})\Vert \boldsymbol{\Theta}_{\boldsymbol{\beta} ^*,j} \Vert^2_2 \leq   \boldsymbol{\Theta}_{\boldsymbol{\beta} ^*,j}^\top {\boldsymbol{\Sigma}}_{\boldsymbol{\beta} ^*} \boldsymbol{\Theta}_{\boldsymbol{\beta} ^*,j} = {\Theta}_{\boldsymbol{\beta} ^*,jj}  = 1/\phi^2_{\boldsymbol{\beta} ^*,j}, $$ then we obtain
\begin{equation}\label{eq.A12}
\lambda_{\min}({\boldsymbol{\Sigma}}_{\boldsymbol{\beta} ^*}) \leq  \phi^2_{\boldsymbol{\beta} ^*,j} \leq \lambda_{\max}({\boldsymbol{\Sigma}}_{\boldsymbol{\beta} ^*}),
 \end{equation}
and
 \begin{equation*} \label{eq.A13}
{\Vert\boldsymbol{\varphi}_{\boldsymbol{\beta} ^*,j} \Vert^2_2 = \Vert \phi^2_{\boldsymbol{\beta} ^*,j} \boldsymbol{\Theta}_{\boldsymbol{\beta} ^*,j}  \Vert^2_2  \leq {\phi^2_{\boldsymbol{\beta} ^*,j}}/{ \lambda_{\min}({\boldsymbol{\Sigma}}_{\boldsymbol{\beta} ^*})}}.
 \end{equation*}
This further implies  
 \begin{equation}\label{eq.A14}
 \max_{1\leq j \leq p}\Vert\boldsymbol{\varphi}_{\boldsymbol{\beta} ^*,j} \Vert_1 
 \leq \max_{1\leq j \leq p}\sqrt{ s_j}\Vert\boldsymbol{\varphi}_{\boldsymbol{\beta} ^*,j} \Vert_2 
 \leq \max_{1\leq j \leq p}\sqrt{s_j {{\Sigma}}_{\boldsymbol{\beta} ^*,jj}/\lambda_{\min}({\boldsymbol{\Sigma}}_{\boldsymbol{\beta} ^*})} = O(\sqrt{s^{**}}),
 \end{equation}
 and
 $$
 \max_{1\leq j \leq p}\Vert \boldsymbol{\varrho}_{\boldsymbol{\beta} ^*,j} \Vert_{\infty} = \max_{1\leq j \leq p} \Vert \mathbf{X}_{\boldsymbol{\beta} ^*}{\boldsymbol{\Theta}}_{{\boldsymbol{\beta} ^*},j}\Vert_{\infty}\cdot  \max_{1\leq j \leq p}\phi^2_{\beta^*,j}  = O_p(\sqrt{s^{**}} K). 
$$

Then by left-producting the diagonal matrix $\mathbf{W}_{\hat{\boldsymbol{\beta}}  }\mathbf{W}^{-1}_{\boldsymbol{\beta} ^*} $ on both sides of \eqref{lemma4.1_01}, we have the $\mathbf{W}_{\hat{\boldsymbol{\beta}}} $-weighted node-wise regression,
\begin{equation*}\label{lemma4.1_02}
  \boldsymbol{X}_{\hat{\boldsymbol{\beta}}  ,(j)} = \mathbf{X}_{\hat{\boldsymbol{\beta}}  ,(-j)} \boldsymbol{\varphi}_{\boldsymbol{\beta} ^*,j} +\mathbf{W}_{\hat{\boldsymbol{\beta}}  }\mathbf{W}^{-1}_{\boldsymbol{\beta} ^*}\boldsymbol{\varrho}_{\boldsymbol{\beta} ^*,j}.
\end{equation*}
By the definition of $\hat{\boldsymbol{\varphi}}_{\hat{\boldsymbol{\beta}  },j}$ in (\ref{op-de}), we have  
\begin{equation*}
\Vert \boldsymbol{X}_{\hat{\boldsymbol{\beta}}  ,(j)} - \mathbf{X}_{\hat{\boldsymbol{\beta}}  ,(-j)} \hat{\boldsymbol{\varphi}}_{\hat{\boldsymbol{\beta}}  ,j} \Vert^2_{2}/n + Q_{\lambda_{j}}(\hat{\boldsymbol{\varphi}}_{\hat{\boldsymbol{\beta}}  ,j}) 
\leq \Vert \boldsymbol{X}_{\hat{\boldsymbol{\beta}}  ,(j)} - \mathbf{X}_{\hat{\boldsymbol{\beta}}  ,(-j)} \boldsymbol{\varphi}_{\boldsymbol{\beta} ^*,j} \Vert^2_{2}/n + Q_{\lambda_{j}}(\boldsymbol{\varphi}_{\boldsymbol{\beta} ^*,j}), 
\end{equation*}
which further indicates that 
\begin{equation}\label{eq.A13old}
\Vert \mathbf{X}_{\hat{\boldsymbol{\beta}}  ,(-j)} (\hat{\boldsymbol{\varphi}}_{\hat{\boldsymbol{\beta}}  ,j} -\boldsymbol{\varphi}_{\boldsymbol{\beta} ^*,j})  \Vert^2_{2}/n + Q_{\lambda_{j}}( \hat{\boldsymbol{\varphi}}_{\hat{\boldsymbol{\beta}}  ,j}) 
\leq 2\boldsymbol{\varrho}_{\boldsymbol{\beta} ^*,j}^\top\mathbf{W}_{\hat{\boldsymbol{\beta}}  }^2\mathbf{W}^{-2}_{\boldsymbol{\beta} ^*}\mathbf{X}_{\boldsymbol{\beta} ^*,(-j)} (\hat{\boldsymbol{\varphi}}_{\hat{\boldsymbol{\beta}}  ,j} -\boldsymbol{\varphi}_{\boldsymbol{\beta} ^*,j} )/n + Q_{\lambda_{j}}(  \boldsymbol{\varphi}_{\boldsymbol{\beta} ^*,j} ). 
\end{equation}
Then by the Holder's inequality, we have
\begin{equation}\label{eq.A17}
\boldsymbol{\varrho}_{\boldsymbol{\beta} ^*,j}^\top\mathbf{W}_{\hat{\boldsymbol{\beta}}  }^2\mathbf{W}^{-2}_{\boldsymbol{\beta} ^*}\mathbf{X}_{\boldsymbol{\beta} ^*,(-j)} (\hat{\boldsymbol{\varphi}}_{\hat{\boldsymbol{\beta}}  ,j} -\boldsymbol{\varphi}_{\boldsymbol{\beta} ^*,j} )/n 
 \leq (\iota_{1j}  + \iota_{2j}) \Vert \hat{\boldsymbol{\varphi}}_{\hat{\boldsymbol{\beta}}  ,j} -\boldsymbol{\varphi}_{\boldsymbol{\beta} ^*,j} \Vert_1,
\end{equation}
where $\iota_{1j}=\Vert \boldsymbol{\varrho}_{\boldsymbol{\beta} ^*,j}^\top \mathbf{X}_{{\boldsymbol{\beta} ^*},(-j)} /n\Vert_{\infty} $ and $\iota_{2j}=\Vert \boldsymbol{\varrho}_{\boldsymbol{\beta} ^*,j}^\top(\mathbf{W}^2_{\hat{\boldsymbol{\beta}}  }\mathbf{W}^{-2}_{\boldsymbol{\beta} ^*}-\mathbf{I}) \mathbf{X}_{{\boldsymbol{\beta} ^*},(-j)} /n\Vert_{\infty}$.
For $\iota_{1j}$ in \eqref{eq.A17},  by Lemma \ref{lemmaA.2}, \eqref{eq.A12}, and \eqref{eq.A14}, we obtain 
\begin{eqnarray}
 \max_{1\leq j \leq p}\iota_{1j}=\max_{1\leq j \leq p}\Vert \boldsymbol{\varrho}_{\boldsymbol{\beta} ^*,j}^\top \mathbf{X}_{{\boldsymbol{\beta} ^*},(-j)} /n\Vert_{\infty}
 &=&\max_{1\leq j \leq p}\Vert \phi^2_{\boldsymbol{\beta} ^*,j}\boldsymbol{\Theta}^\top_{\boldsymbol{\beta} ^*,j} (\mathbf{X}^\top_{{\boldsymbol{\beta} ^*}}\mathbf{X}_{{\boldsymbol{\beta} ^*},(-j)}/n-\boldsymbol{\Sigma}_{\boldsymbol{\beta} ^*,-j}) \Vert_{\infty}\nonumber\\
 &=& O_p(\sqrt{s^{**}\ln p/n}), \label{eq.A18}
 \end{eqnarray}
 where $\boldsymbol{\Sigma}_{\boldsymbol{\beta} ^*,-j}={\sf}[\mathbf{X}^\top_{{\boldsymbol{\beta} ^*}}\mathbf{X}_{{\boldsymbol{\beta} ^*},(-j)}/n]$.
For $\iota_{2j}$ in \eqref{eq.A17},
by the Holder's inequality we have 
\begin{eqnarray}
\max_{1\leq j \leq p}\iota_{2j}
&=&\max_{1\leq j \leq p}\Vert \boldsymbol{\varrho}_{\boldsymbol{\beta} ^*,j}^\top(\mathbf{W}^2_{\hat{\boldsymbol{\beta}}  }\mathbf{W}^{-2}_{\boldsymbol{\beta} ^*}-\mathbf{I}) \mathbf{X}_{{\boldsymbol{\beta} ^*},(-j)} /n\Vert_{\infty}\\
&=&\max_{1\leq j \leq p}\max_{1\leq q\leq p, q\neq j} \frac{1}{n}\left|\sum_{i=1}^n\frac{w^2_{\hat{\boldsymbol{\beta}} ,i}-w^2_{\boldsymbol{\beta} ^*,i}}{w^2_{\boldsymbol{\beta} ^*,i}}{\varrho}_{\boldsymbol{\beta} ^*,ji}x_{\boldsymbol{\beta} ^*,iq}\right|\nonumber\\
&\leq& \max_{1\leq j \leq p}\max_{1\leq q\leq p, q\neq j}\max_{1\leq i\leq n}|{\varrho}_{\boldsymbol{\beta} ^*,ji}x_{\boldsymbol{\beta} ^*,il} |\cdot \frac{1}{n}\sum_{i=1}^n{|w^2_{\hat{\boldsymbol{\beta}} ,i}-w^2_{\boldsymbol{\beta} ^*,i}|},\nonumber\\
&=&O_p(\sqrt{s^{**}}K^2 \cdot sK\lambda ) = O_p( s\sqrt{s^{**}} K^{3} \lambda), \label{eq.A19}
\end{eqnarray}
where the last step follows from \eqref{eq.A14} the  results in Proposition \ref{prop1} and Lemma \ref{lem1}. 




Denote by the population covariance matrix ${\boldsymbol\Sigma}_{-j,-j} = {\sf E} [ \mathbf{X}^\top_{(-j)} \mathbf{X}_{(-j)}/n ]$ and the sample covariance $\hat{\boldsymbol\Sigma}_{-j,-j} = \mathbf{X}^\top_{(-j)} \mathbf{X}_{(-j)}/n$, respectively. Note further that 
$$\lambda_{\min}(\boldsymbol{\Sigma}) \leq \lambda_{\min}({\boldsymbol\Sigma}_{-j,-j}) \quad \text{and} \quad \Vert {\boldsymbol\Sigma}_{-j,-j} -   \hat{\boldsymbol\Sigma}_{-j,-j}\Vert_{\infty} \leq \Vert \boldsymbol{\Sigma}  -   \hat{\boldsymbol{\Sigma}} \Vert_{\infty},$$
then we have 
\begin{equation}\label{eq.A20}
\begin{aligned}
\Vert \mathbf{X}_{\hat{\boldsymbol{\beta}}  ,(-j)} (\hat{\boldsymbol{\varphi}}_{\hat{\boldsymbol{\beta}}  ,j} -\boldsymbol{\varphi}_{\boldsymbol{\beta} ^*,j})  \Vert^2_{2}/n  
& \geq  \min\{\tau,1-\tau\} \cdot \Vert \mathbf{X}_{(-j)} (\hat{\boldsymbol{\varphi}}_{\hat{\boldsymbol{\beta}}  ,j} -\boldsymbol{\varphi}_{\boldsymbol{\beta} ^*,j})  \Vert^2_{2}/n \\
& \geq \min\{\tau,1-\tau\} \left(\lambda_{\min}({\boldsymbol\Sigma}_{-j,-j}) \Vert \hat{\boldsymbol{\varphi}}_{\hat{\boldsymbol{\beta}}  ,j} -\boldsymbol{\varphi}_{\boldsymbol{\beta} ^*,j}  \Vert^2_2 - \Vert {\boldsymbol\Sigma}_{-j,-j} -   \hat{\boldsymbol\Sigma}_{-j,-j}\Vert_{\infty}\Vert \hat{\boldsymbol{\varphi}}_{\hat{\boldsymbol{\beta}}  ,j} -\boldsymbol{\varphi}_{\boldsymbol{\beta} ^*,j}  \Vert^2_1 \right) \\
& \geq  \min\{\tau,1-\tau\} \left(\lambda_{\min}(\boldsymbol{\Sigma}) \Vert \hat{\boldsymbol{\varphi}}_{\hat{\boldsymbol{\beta}}  ,j} -\boldsymbol{\varphi}_{\boldsymbol{\beta} ^*,j}  \Vert^2_2 - \Vert \boldsymbol{\Sigma}  -   \hat{\boldsymbol{\Sigma}} \Vert_{\infty}\Vert \hat{\boldsymbol{\varphi}}_{\hat{\boldsymbol{\beta}}  ,j} -\boldsymbol{\varphi}_{\boldsymbol{\beta} ^*,j}  \Vert^2_1 \right).
\end{aligned}
\end{equation}
Since
$ \Vert \boldsymbol{\varphi}_{\boldsymbol{\beta} ^*,j} \Vert_1 \leq R_j$, we have 
$\Vert \hat{\boldsymbol{\varphi}}_{\hat{\boldsymbol{\beta}}  ,j} -\boldsymbol{\varphi}_{\boldsymbol{\beta} ^*,j}  \Vert_1 \leq \Vert \hat{\boldsymbol{\varphi}}_{\hat{\boldsymbol{\beta}}  ,j}  \Vert_1  +\Vert \boldsymbol{\varphi}_{\boldsymbol{\beta} ^*,j}  \Vert_1 \leq 2 R_j$.
 Combining  \eqref{eq.A13old} with \eqref{eq.A17} and \eqref{eq.A20}, we can derive that
\begin{equation*}
\begin{aligned}
0 &\leq \min\{\tau,1-\tau\}  \lambda_{\min}(\boldsymbol{\Sigma}) \Vert \hat{\boldsymbol{\varphi}}_{\hat{\boldsymbol{\beta}}  ,j} -\boldsymbol{\varphi}_{\boldsymbol{\beta} ^*,j}  \Vert^2_2 \\
& \leq  \left\{ \iota_{1j} +\iota_{2j} + 2\min\{\tau,1-\tau\} R_j \Vert \boldsymbol{\Sigma}  -   \hat{\boldsymbol{\Sigma}} \Vert_{\infty}   \right\} \Vert \hat{\boldsymbol{\varphi}}_{\hat{\boldsymbol{\beta}}  ,j} -\boldsymbol{\varphi}_{\boldsymbol{\beta} ^*,j} \Vert_1  + Q_{\lambda_{j}}(  \boldsymbol{\varphi}_{\boldsymbol{\beta} ^*,j} ) - Q_{\lambda_{j}}( \hat{\boldsymbol{\varphi}}_{\hat{\boldsymbol{\beta}}  ,j}),
\end{aligned}
\end{equation*}
where $\iota_{1j}$ and $\iota_{2j}$ are defined in \eqref{eq.A17}.
Under the event   
$$ {\cal E}_j=
 \{\lambda_j\geq 4(\iota_{1j}+\iota_{2j}) \} \cap\{ \lambda_j\geq 8 \min\{\tau,1-\tau\} \Vert \boldsymbol{\Sigma}  -   \hat{\boldsymbol{\Sigma}} \Vert_{\infty} R_j\},
$$
by invoking Lemma \ref{lux_amen_convex} and  applying the sub-addictive property of the amenable regularizer in Proposition \ref{prop1}, it yields that
$$
\min\{\tau,1-\tau\}  \lambda_{\min}(\boldsymbol{\Sigma}) \Vert \hat{\boldsymbol{\varphi}}_{\hat{\boldsymbol{\beta}}  ,j} -\boldsymbol{\varphi}_{\boldsymbol{\beta} ^*,j}  \Vert^2_2 \leq \frac{3}{2}Q_{\lambda_{j}}(  \boldsymbol{\varphi}_{\boldsymbol{\beta} ^*,j} ) - \frac{1}{2}Q_{\lambda_{j}}( \hat{\boldsymbol{\varphi}}_{\hat{\boldsymbol{\beta}}  ,j}) + \frac{\mu}{2}\Vert \hat{\boldsymbol{\varphi}}_{\hat{\boldsymbol{\beta}}  ,j} -\boldsymbol{\varphi}_{\boldsymbol{\beta} ^*,j} \Vert^2_2.
$$
Since 
$$\min\{\tau,1-\tau\}  \lambda_{\min}(\boldsymbol{\Sigma}) \geq 3\mu/4 > \mu/2, $$
then 
$$
0 \leq \left(\min\{\tau,1-\tau\}  \lambda_{\min}(\boldsymbol{\Sigma}) - \mu/2 \right)\Vert \hat{\boldsymbol{\varphi}}_{\hat{\boldsymbol{\beta}}  ,j} -\boldsymbol{\varphi}_{\boldsymbol{\beta} ^*,j} \Vert^2_2 \leq \frac{3}{2}Q_{\lambda_{j}}(  \boldsymbol{\varphi}_{\boldsymbol{\beta} ^*,j} ) - \frac{1}{2}Q_{\lambda_{j}}( \hat{\boldsymbol{\varphi}}_{\hat{\boldsymbol{\beta}}  ,j}).
$$
Letting $s_j=\Vert \boldsymbol{\varphi}_{\boldsymbol{\beta} ^*,j} \Vert_0$,
applying the conclusion in Lemma \ref{lux_amen_L1}, and following the similar argument in the proof of Proposition \ref{prop1}, we obtain the $l_2$ bound
\begin{equation}
\label{eq.A21}
\Vert \hat{\boldsymbol{\varphi}}_{\hat{\boldsymbol{\beta}}  ,j} -\boldsymbol{\varphi}_{\boldsymbol{\beta} ^*,j} \Vert_2 \leq \frac{6\sqrt{s_j}\lambda_j }{4\min\{\tau,1-\tau\}\lambda_{\min}(\boldsymbol{\Sigma}) - \mu },
\end{equation}
also the $l_1$ bound 
\begin{equation}\label{eq.A22}
\Vert \hat{\boldsymbol{\varphi}}_{\hat{\boldsymbol{\beta}}  ,j} -\boldsymbol{\varphi}_{\boldsymbol{\beta} ^*,j} \Vert_1 \leq \frac{24 {s_j}\lambda_j }{4\min\{\tau,1-\tau\}\lambda_{\min}(\boldsymbol{\Sigma}) - \mu },
\end{equation}
and finally the prediction error bound
\begin{equation}\label{eq.A23}
\Vert \mathbf{X}_{(-j)} (\hat{\boldsymbol{\varphi}}_{\hat{\boldsymbol{\beta}}  ,j} -\boldsymbol{\varphi}_{\boldsymbol{\beta} ^*,j})  \Vert^2_2 /n \leq \frac{s_j\lambda_j^2}{\min\{\tau,1-\tau\}} \left ( \frac{12}{4\min\{\tau,1-\tau\}\lambda_{\min}(\boldsymbol{\Sigma}) - \mu} + \frac{36\mu}{(4\min\{\tau,1-\tau\}\lambda_{\min}(\boldsymbol{\Sigma}) - \mu)^2}\right).
\end{equation}

Lastly, we can prove 
$\Pr (\bigcap_{j=1}^p {\cal E}_j)\to 1$ as $n\to\infty$ by \eqref{eq.A18} and \eqref{eq.A19} and  choosing $\lambda_j$ such that  $\lambda_j\geq c_8(\max_jR_j\sqrt{\ln p/n})\vee(s\sqrt{s^{**}}K^3 \lambda))$ for $j\in\{1,\ldots, p\}$. 
By replacing $s_j$ and $\lambda_j$ in \eqref{eq.A21}--\eqref{eq.A23} with $s^{**}$ and $\lambda^{**}$, respectively, we completes the proof of (i).

\medskip

(ii) Recall that 
 $$
  \hat{\phi}^2_{\hat{\boldsymbol{\beta}}  , j}  = \boldsymbol{X}_{\hat{\boldsymbol{\beta}}  ,(j)}^\top( \boldsymbol{X}_{\hat{\boldsymbol{\beta}}  ,(j)} - \mathbf{X}_{\hat{\boldsymbol{\beta}}  ,(-j)} \hat{\boldsymbol{\varphi}}_{\hat{\boldsymbol{\beta}}  ,j} )/n.
 $$
Since 
 $$
\boldsymbol{X}_{\hat{\boldsymbol{\beta}},(j)} = \mathbf{W}_{\hat{\boldsymbol{\beta}}  } \mathbf{W}^{-1}_{\boldsymbol{\beta} ^*} \boldsymbol{X}_{{\boldsymbol{\beta} ^*},(j)} = \mathbf{W}_{\hat{\boldsymbol{\beta}}  } \mathbf{W}^{-1}_{\boldsymbol{\beta} ^*} (\mathbf{X}_{\boldsymbol{\beta} ^*,(-j)} \boldsymbol{\varphi}_{\boldsymbol{\beta} ^*,j} +\boldsymbol{\varrho}_{\boldsymbol{\beta} ^*,j}),
 $$
and 
 $$
( \boldsymbol{X}_{\hat{\boldsymbol{\beta}}  ,(j)} - \mathbf{X}_{\hat{\boldsymbol{\beta}}  ,(-j)} \hat{\boldsymbol{\varphi}}_{\hat{\boldsymbol{\beta}}  , j} ) = \mathbf{W}_{\hat{\boldsymbol{\beta}}  } \mathbf{W}^{-1}_{\boldsymbol{\beta} ^*} ( \boldsymbol{X}_{\boldsymbol{\beta} ^*,(j)} - \mathbf{X}_{\boldsymbol{\beta} ^*,(-j)}  \hat{\boldsymbol{\varphi}}_{\hat{\boldsymbol{\beta}}  , j}    )
=  \mathbf{W}_{\hat{\boldsymbol{\beta}}  } \mathbf{W}^{-1}_{\boldsymbol{\beta} ^*} ( \boldsymbol{\varrho}_{\boldsymbol{\beta} ^*,j} + \mathbf{X}_{\boldsymbol{\beta} ^*,(-j)}  (  \boldsymbol{\varphi}_{\boldsymbol{\beta} ^*,j} -    \hat{\boldsymbol{\varphi}}_{\hat{\boldsymbol{\beta}}  ,j} )   ),
 $$
then we can write
\begin{eqnarray*}
\hat{\phi}^2_{\hat{\boldsymbol{\beta}}  , j} - \phi^2_{\boldsymbol{\beta} ^*,j} 
&=& \left[  \boldsymbol{X}^\top_{{\boldsymbol{\beta} ^*},(j)} ( \boldsymbol{X}_{\boldsymbol{\beta} ^*,(j)} - \mathbf{X}_{\boldsymbol{\beta} ^*,(-j)}  \hat{\boldsymbol{\varphi}}_{\hat{\boldsymbol{\beta}}  , j} )/n -  \phi^2_{\boldsymbol{\beta} ^*,j}  \right] \\
&&+ \left[  \boldsymbol{X}^\top_{{\boldsymbol{\beta} ^*},(j)}( \mathbf{W}^2_{\hat{\boldsymbol{\beta}}  } \mathbf{W}^{-2}_{\boldsymbol{\beta} ^*}  -\mathbf{I}   ) ( \boldsymbol{X}_{\boldsymbol{\beta} ^*,(j)} - \mathbf{X}_{\boldsymbol{\beta} ^*,(-j)}  \hat{\boldsymbol{\varphi}}_{\hat{\boldsymbol{\beta}}  , j} )/n   \right] \\
&=:& \iota_{3j}+\iota_{4j}.
\end{eqnarray*}

For $\iota_{3j}$, we have
\begin{eqnarray}
&&\left|\boldsymbol{X}^\top_{{\boldsymbol{\beta} ^*},(j)} ( \boldsymbol{X}_{\boldsymbol{\beta} ^*,(j)} - \mathbf{X}_{\boldsymbol{\beta} ^*,(-j)}  \hat{\boldsymbol{\varphi}}_{\hat{\boldsymbol{\beta}}  , j} )/n -  \phi^2_{\boldsymbol{\beta} ^*,j}  \right| \\
&\leq& \left|  \boldsymbol{\varrho}^\top_{\boldsymbol{\beta} ^*,j} \boldsymbol{\varrho}_{\boldsymbol{\beta} ^*,j}/n - \phi^2_{\boldsymbol{\beta} ^*,j}  \right| + \left|  \boldsymbol{\varrho}^\top_{\boldsymbol{\beta} ^*,j} \mathbf{X}_{\boldsymbol{\beta} ^*,(-j)}  (  \boldsymbol{\varphi}_{\boldsymbol{\beta} ^*,j} -    \hat{\boldsymbol{\varphi}}_{\hat{\boldsymbol{\beta}}  ,j} ) /n  \right| \nonumber \\ 
&& \quad + \left|  \boldsymbol{\varrho}^\top_{\boldsymbol{\beta} ^*,j} \mathbf{X}_{\boldsymbol{\beta} ^*,(-j)}   \boldsymbol{\varphi}_{\boldsymbol{\beta} ^*,j} /n  \right| + \left| \boldsymbol{\varphi}^\top_{\boldsymbol{\beta} ^*,j} \mathbf{X}^\top_{\boldsymbol{\beta} ^*,(-j)} \mathbf{X}_{\boldsymbol{\beta} ^*,(-j)}  (  \boldsymbol{\varphi}_{\boldsymbol{\beta} ^*,j} -    \hat{\boldsymbol{\varphi}}_{\hat{\boldsymbol{\beta}}  ,j} ) /n  \right| \nonumber\\ 
&=:& \iota_{3j,1}+\iota_{3j,2}+\iota_{3j,3}+\iota_{3j,4}.\nonumber
\end{eqnarray}
By \eqref{eq.A14}, we can prove 
\begin{equation}\label{eq.A24}
\max_{1\leq j \leq p}\iota_{3j,1} =\max_{1\leq j \leq p}\left| \phi^4_{\boldsymbol{\beta} ^*,j}\boldsymbol{\Theta}_{\boldsymbol{\beta} ^*,j}^\top (\mathbf{X}_{{\boldsymbol{\beta} ^*}}^\top\mathbf{X}_{{\boldsymbol{\beta} ^*}}/n-\boldsymbol{\Sigma}_{\boldsymbol{\beta} ^*}) \boldsymbol{\Theta}_{\boldsymbol{\beta} ^*,j} \right|
= O_p(s^{**}\sqrt{\ln p / n }). 
\end{equation}

Besides, by Holder's inequality and \eqref{eq.A18}, we have 
 \begin{eqnarray}
\max_{1\leq j \leq p}\iota_{3j,2} &\leq& \max_{1\leq j\leq p}\Vert \boldsymbol{\varrho}^\top_{\boldsymbol{\beta} ^*,j} \mathbf{X}_{\boldsymbol{\beta} ^*,(-j)} /n \Vert_{\infty}
\cdot\max_{1\leq j \leq p} \Vert  \boldsymbol{\varphi}_{\boldsymbol{\beta} ^*,j} -    \hat{\boldsymbol{\varphi}}_{\hat{\boldsymbol{\beta}}  ,j} \Vert_1 \nonumber\\
&=& {  O_p( \sqrt{s^{**}\ln p/n}) O_p(s^{**} \lambda^{**}) = O_p((s^{**})^{3/2}\lambda^{**}\sqrt{\ln p/n})},
 \end{eqnarray}
and
 \begin{equation}
\max_{1\leq j \leq p}\iota_{3j,3} 
\leq \max_{1\leq j \leq p} { \Vert \boldsymbol{\varrho}^\top_{\boldsymbol{\beta} ^*,j} \mathbf{X}_{\boldsymbol{\beta} ^*,(-j)} /n\Vert_{\infty}}
 \cdot \max_{1\leq j \leq p}\Vert\boldsymbol{\varphi}_{\boldsymbol{\beta} ^*,j} \Vert_1  =  { O_p(s^{**}\sqrt{\ln p/n}).}
 \end{equation}
Note further that
\begin{eqnarray*}
&&\Vert  \mathbf{X}_{\boldsymbol{\beta} ^*,(-j)} \boldsymbol{\varphi}_{\boldsymbol{\beta} ^*,j}\Vert^2_2/n\\
&\leq& \left|  \Vert \mathbf{X}_{\boldsymbol{\beta} ^*,(-j)} \boldsymbol{\varphi}_{\boldsymbol{\beta} ^*,j} \Vert^2_2/n - \boldsymbol{\varphi}^\top_{\boldsymbol{\beta} ^*,j} {\sf E}\left[ \mathbf{X}^\top_{\boldsymbol{\beta} ^*,(-j)}  \mathbf{X}_{\boldsymbol{\beta} ^*,(-j)}/n\right] \boldsymbol{\varphi}_{\boldsymbol{\beta} ^*,j}  \right| + \boldsymbol{\varphi}^\top_{\boldsymbol{\beta} ^*,j} {\sf E}\left[ \mathbf{X}^\top_{\boldsymbol{\beta} ^*,(-j)}  \mathbf{X}_{\boldsymbol{\beta} ^*,(-j)}/n\right] \boldsymbol{\varphi}_{\boldsymbol{\beta} ^*,j} \\
& \leq & \lVert \mathbf{X}^\top_{\boldsymbol{\beta} ^*,(-j)}  \mathbf{X}_{\boldsymbol{\beta} ^*,(-j)}/n -{\sf E}\left[ \mathbf{X}^\top_{\boldsymbol{\beta} ^*,(-j)}  \mathbf{X}_{\boldsymbol{\beta} ^*,(-j)}/n\right]  \rVert_{\infty} \Vert\boldsymbol{\varphi}_{\boldsymbol{\beta} ^*,j} \Vert^2_1 + \lambda_{\max}(\Sigma_{\boldsymbol{\beta}^*})\Vert\boldsymbol{\varphi}_{\boldsymbol{\beta} ^*,j} \Vert^2_2 \\ 
& \leq &  O_{p}(\sqrt{\ln p/n})O(s_{j}) + O(1) = O(1),
\end{eqnarray*}
and by the Cauchy–Schwarz inequality,
\begin{eqnarray}
\max_{1\leq j \leq p}\iota_{3j,4} 
&\leq&\max_{1\leq j \leq p} \Vert  \mathbf{X}_{\boldsymbol{\beta} ^*,(-j)} \boldsymbol{\varphi}_{\boldsymbol{\beta} ^*,j}/\sqrt{n}\Vert_2 
\cdot \max_{1\leq j \leq p}\Vert \mathbf{X}_{\boldsymbol{\beta} ^*,(-j)}  (  \boldsymbol{\varphi}_{\boldsymbol{\beta} ^*,j} -    \hat{\boldsymbol{\varphi}}_{\hat{\boldsymbol{\beta}}  ,j} )/\sqrt{n}\Vert_2 \nonumber \\
&\leq&\max_{1\leq j \leq p} \Vert  \mathbf{X}_{\boldsymbol{\beta} ^*,(-j)} \boldsymbol{\varphi}_{\boldsymbol{\beta} ^*,j}/\sqrt{n}\Vert_2 
\cdot  \sqrt{\max\{ \tau,1-\tau\}} \max_{1\leq j \leq p}\Vert \mathbf{X}_{(-j)} (\hat{\boldsymbol{\varphi}}_{\hat{\boldsymbol{\beta}}  ,j} -\boldsymbol{\varphi}_{\boldsymbol{\beta} ^*,j})/\sqrt{n}  \Vert_2  \nonumber\\
&=& O_p(\sqrt{s^{**}}\lambda^{**}). \label{eq.A27}
\end{eqnarray}
Thus, combining \eqref{eq.A24}--\eqref{eq.A27} we prove that
 $$
\max_{1\leq j \leq p}|\iota_{3j}| = O_p(s^{**}\sqrt{\ln p / n }) + O_p((s^{**})^{3/2}\lambda^{**}\sqrt{\ln p/n}) + {{  O_p(s^{**}\sqrt{\ln p/n})} + O_p(\sqrt{s^{**}}\lambda^{**})}  = O_p(\sqrt{s^{**}}\lambda^{**}). 
 $$

Now we turn to $\iota_{4j}$. Similar to  $\iota_{2j}$, we have   
\begin{equation*}
\begin{aligned}
|\iota_{4j}| & \leq \left|  \boldsymbol{X}^\top_{{\boldsymbol{\beta} ^*},(j)}( \mathbf{W}^2_{\hat{\boldsymbol{\beta}}  } \mathbf{W}^{-2}_{\boldsymbol{\beta} ^*}  -\mathbf{I}   ) ( \boldsymbol{X}_{\boldsymbol{\beta} ^*,(j)} - \mathbf{X}_{\boldsymbol{\beta} ^*,(-j)}  {\boldsymbol{\varphi}}_{ {\boldsymbol{\beta}^*}  , j} )/n   \right| + \left|  \boldsymbol{X}^\top_{{\boldsymbol{\beta} ^*},(j)}( \mathbf{W}^2_{\hat{\boldsymbol{\beta}}  } \mathbf{W}^{-2}_{\boldsymbol{\beta} ^*}  -\mathbf{I}   )  \mathbf{X}_{\boldsymbol{\beta} ^*,(-j)}  ( \hat{\boldsymbol{\varphi}}_{ \hat{\boldsymbol{\beta}}  , j} -
 {\boldsymbol{\varphi}}_{ {\boldsymbol{\beta}^*}  , j} )/n   \right| \\
 & \leq \max_{1\leq q\leq p, q\neq j} \frac{1}{n}\left|\sum_{i=1}^n\frac{w^2_{\hat{\boldsymbol{\beta}} ,i}-w^2_{\boldsymbol{\beta} ^*,i}}{w^2_{\boldsymbol{\beta} ^*,i}}{\varrho}_{\boldsymbol{\beta} ^*,j,i}x_{\boldsymbol{\beta} ^*,ij}\right| +\max_{1\leq q\leq p, q\neq j} \frac{1}{n}\left|\sum_{i=1}^n\frac{w^2_{\hat{\boldsymbol{\beta}} ,i}-w^2_{\boldsymbol{\beta} ^*,i}}{w^2_{\boldsymbol{\beta} ^*,i}} |x_{\boldsymbol{\beta} ^*,ij} ||\boldsymbol{X}_{\boldsymbol{\beta}^*,(-j),i}(\hat{\boldsymbol{\varphi}}_{ \hat{\boldsymbol{\beta}}  , j} -
 {\boldsymbol{\varphi}}_{ {\boldsymbol{\beta}^*}  , j}   )|\right|.
\end{aligned}
\end{equation*}
{Since $\Vert \boldsymbol{X}_{{\boldsymbol{\beta} ^*},(j)}  \Vert_{\infty} \leq   \Vert \boldsymbol{X}_{(j)}  \Vert_{\infty} \leq \Vert \boldsymbol{X}   \Vert_{\infty}   $ and 
 $$
\max_{1\leq j \leq p}\Vert  \mathbf{X}_{\boldsymbol{\beta} ^*,(-j)}  (\hat{\boldsymbol{\varphi}}_{ \hat{\boldsymbol{\beta}}  , j} -
 {\boldsymbol{\varphi}}_{ {\boldsymbol{\beta}^*}  , j}   ) \Vert_{\infty} 
 \leq \max_{1\leq j \leq p} \Vert \boldsymbol{X}_{{\boldsymbol{\beta} ^*},(-j)}  \Vert_{\infty}  
 \cdot \max_{1\leq j \leq p}\Vert  {\boldsymbol{\varphi}}_{{\boldsymbol{\beta} ^*}, j} -  \hat{\boldsymbol{\varphi}}_{\hat{\boldsymbol{\beta}}  , j} \Vert_{1}  =  O_p(Ks^{**}\lambda^{**}),
 $$
then by the Holder's inequality and \eqref{eq.A19}, we get}
\begin{equation}\label{eq.A25}
\max_{1\leq j \leq p}|\iota_{4j}| \leq O_p(s\sqrt{s^{**}}K^3\lambda) +O_p(Ks\lambda)O_p(K) O_p(Ks^{**}\lambda^{**})   = O_p(ss^{**}K^3\lambda).
\end{equation}
Thus by \eqref{eq.A24} and \eqref{eq.A25} we prove that 
\begin{equation}\label{lemma4.1_06}
\max_{1\leq j \leq p}|\hat{\phi}^2_{\hat{\boldsymbol{\beta}}  , j} - \phi^2_{\boldsymbol{\beta} ^*,j}| =  O_p(\sqrt{s^{**}}\lambda^{**}) +   O_p( ss^{**}K^3\lambda) = O_p(\sqrt{s^{**}}\lambda^{**}). 
\end{equation}
Then by \eqref{eq.A12} we have
\begin{equation}\label{lemma4.1_07}
\max_{1\leq j \leq p}|1/\hat{\phi}^2_{\hat{\boldsymbol{\beta}}  , j} - 1/\phi^2_{\boldsymbol{\beta} ^*,j}| = O_p(\sqrt{s^{**}}\lambda^{**}). 
\end{equation}

Finally, using   \eqref{lemma4.1_06}, \eqref{lemma4.1_07}, we have  
\begin{eqnarray*}
\max_{1\leq j \leq p}  \left\lVert \hat{\boldsymbol{\Theta}}_{\hat{\boldsymbol{\beta}}  ,j} - \boldsymbol{\Theta}_{\boldsymbol{\beta} ^*,j} \right\rVert _{1}
&= & \max_{1\leq j \leq p}\left\lVert \hat{\boldsymbol{\Phi}}_{\hat{\boldsymbol{\beta}}  ,j}/\hat{\phi}^2_{\hat{\boldsymbol{\beta}}  ,j} - {\boldsymbol{\Phi}}_{\boldsymbol{\beta} ^*,j}/\phi^2_{\boldsymbol{\beta} ^*,j} \right\rVert_{1}\\
 & \leq & \max_{1\leq j \leq p}\left(\Vert \hat{\boldsymbol{\varphi}}_{\hat{\boldsymbol{\beta}}  ,j} - \boldsymbol{\varphi}_{\boldsymbol{\beta} ^*,j} \Vert_{1} /\hat{\phi}^2_{\hat{\boldsymbol{\beta}}  ,j}\right) +  \max_{1\leq j \leq p}\left(\Vert \boldsymbol{\varphi}_{\boldsymbol{\beta} ^*,j} \Vert_{1} (1/\hat{\phi}^2_{\hat{\boldsymbol{\beta}}  ,j} -1/\phi^2_{\boldsymbol{\beta} ^*,j})\right) \\
 & = &  O_{p}(s^{**}\lambda^{**}) + O(\sqrt{s^{**}}) O_p(\sqrt{s^{**}}\lambda^{**}) = O_{p}(s^{**}\lambda^{**}).
\end{eqnarray*}
Analogously,
\begin{eqnarray*}
\max_{1\leq j \leq p}  \Vert \hat{\boldsymbol{\Theta}}_{\hat{\boldsymbol{\beta}}  ,j} - \boldsymbol{\Theta}_{\boldsymbol{\beta} ^*,j}  \Vert _{2} 
&\leq& \max_{1\leq j \leq p}\left(\Vert \hat{\boldsymbol{\varphi}}_{\hat{\boldsymbol{\beta}}  ,j} - \boldsymbol{\varphi}_{\boldsymbol{\beta} ^*,j} \Vert_{2} /\hat{\phi}^2_{\hat{\boldsymbol{\beta}}  ,j}\right) +  \max_{1\leq j \leq p}\left(\Vert \boldsymbol{\varphi}_{\boldsymbol{\beta} ^*,j} \Vert_{2} (1/\hat{\phi}^2_{\hat{\boldsymbol{\beta}}  ,j} -1/\phi^2_{\boldsymbol{\beta} ^*,j})\right)\\
  &= &O_{p}(\sqrt{s^{**}}\lambda^{**}) +O(1)O_p(\sqrt{s^{**}}\lambda^{**})= O_{p}(\sqrt{s^{**}}\lambda^{**}).
\end{eqnarray*}

Now we complete our proof for Theorem \ref{thm1}.
\end{proof}

\subsection{Proof of Theorem \ref{thm_add}}


\begin{proof}[\textbf{\upshape Proof of Theorem \ref{thm_add}}]
 First, we focus on $  {\boldsymbol{ \Delta}}^{(1)} $. By the Holder's inequality, we have
  $$
  {  \Vert{\boldsymbol{ \Delta}}^{(1)} \Vert_{\infty} \leq \Vert \hat{\boldsymbol{\Theta}}_{\hat{\boldsymbol{\beta}  }} \mathbf{X}^\top \Vert_{\infty} \Vert \mathbf{W}_{\hat{\boldsymbol{\beta}  }}^2\boldsymbol{\epsilon}/n  - \mathbf{W}_{\boldsymbol{\beta}  ^*}^2\boldsymbol{\epsilon}/n \Vert_1.}
 $$
Note that 
 \begin{equation}\label{pf_thad_01}
  \begin{aligned}
\Vert\hat{\boldsymbol{\Theta}}_{\hat{\boldsymbol{\beta}}  } \mathbf{X}^\top\Vert_{\infty} = O(  \max_{1\leq j \leq p} \Vert \mathbf{X}_{\boldsymbol{\beta} ^*}\hat{\boldsymbol{\Theta}}_{\hat{\boldsymbol{\beta}}  ,j}\Vert_{\infty}) &\leq O\left( \max_{1\leq j \leq p}\Vert \mathbf{X}_{\boldsymbol{\beta} ^*}{\boldsymbol{\Theta}}_{{\boldsymbol{\beta} ^*},j}\Vert_{\infty}  +  \max_{1\leq j \leq p}\Vert \mathbf{X}_{\boldsymbol{\beta} ^*}(\hat{\boldsymbol{\Theta}}_{\hat{\boldsymbol{\beta}}  ,j}- {\boldsymbol{\Theta}}_{{\boldsymbol{\beta} ^*},j})\Vert_{\infty}  \right)\\ 
  & \leq {  O_p(\sqrt{s^{**}} K  ) }+ O_p(K )O_p(s^{**}\lambda^{**}) = { O_p(\sqrt{s^{**}} K  ) }.
  \end{aligned}
   \end{equation}
Moreover, by the arguments in Lemma \ref{lem1}, 
 we have
\begin{equation}\label{pf_thad_02}
\begin{aligned}
\Vert \mathbf{W}_{\hat{\boldsymbol{\beta}  }}^2\boldsymbol{\epsilon}/n  - \mathbf{W}_{\boldsymbol{\beta}  ^*}^2\boldsymbol{\epsilon}/n \Vert_1 & \leq  \frac{1}{n}\sum_{i=1}^n\left|w^2_{\hat{\boldsymbol{\beta}} ,i}-w^2_{\boldsymbol{\beta} ^*,i}\right||\boldsymbol{X}_{i}^\top (\hat{\boldsymbol{\beta}} -\boldsymbol{\beta} ^*)|  \\
& \leq O_p(Ks\lambda)\cdot  \Vert \mathbf{X} \Vert_{\infty}  \Vert \boldsymbol{\beta} ^* - \hat{\boldsymbol{\beta}}   \Vert_1
= O_p(K^2 s^2\lambda^2).
\end{aligned}
\end{equation}
Combine \eqref{pf_thad_01} with \eqref{pf_thad_02}, then it follows
\begin{equation*}\label{pf_thad_delta1}
    \Vert\hat{\boldsymbol{\Theta}}_{\hat{\boldsymbol{\beta}}  } \mathbf{X}^\top\Vert_{\infty} \leq  {  O_p(\sqrt{s^{**}} K  ) } O_p(K^2 s^2\lambda^2) = O_p( s\lambda\lambda^{**} ).
\end{equation*}

Next, we turn to ${\boldsymbol{ \Delta}}^{(2)}$. Recall the first order necessary condition for the optimization problem (\ref{op-de}),  that is, 
  \begin{equation*}
  \mathbf{X}_{\hat{\boldsymbol{\beta}  },(-j)}^\top( \boldsymbol{X}_{\hat{\boldsymbol{\beta}  },(j)} - \mathbf{X}_{\hat{\boldsymbol{\beta}  },(-j)} \hat{\boldsymbol{\varphi}}_{\hat{\boldsymbol{\beta}  },j} )+  \nabla Q_{\lambda_{j}}(\hat{\boldsymbol{\varphi}}_{\hat{\boldsymbol{\beta}  },j}) = \mathbf{0}.
  \end{equation*} 
Hence we get 
\begin{equation*}
  \hat{\phi}^2_{\hat{\boldsymbol{\beta}  },j} =  \Vert \boldsymbol{X}_{\hat{\boldsymbol{\beta}  },(j)} - \mathbf{X}_{\hat{\boldsymbol{\beta}  },(-j)} \hat{\boldsymbol{\varphi}}_{\hat{\boldsymbol{\beta}  },j} \Vert^2_{2}/n + \hat{\boldsymbol{\varphi}}^\top_{\hat{\boldsymbol{\beta}  },j} \nabla Q_{\lambda_{j}}(\hat{\boldsymbol{\varphi}}_{\hat{\boldsymbol{\beta}  },j})
  = \boldsymbol{X}_{\hat{\boldsymbol{\beta}  },(j)}^\top( \boldsymbol{X}_{\hat{\boldsymbol{\beta}  },(j)} - \mathbf{X}_{\hat{\boldsymbol{\beta}  },(-j)} \hat{\boldsymbol{\varphi}}_{\hat{\boldsymbol{\beta}  },j} )/n
=\boldsymbol{X}_{\hat{\boldsymbol{\beta}  },(j)}^\top\mathbf{X}_{\hat{\boldsymbol{\beta}  }}\hat{\boldsymbol{\Phi}}_{\hat{\boldsymbol{\beta}  },j}/n,
\end{equation*}  
 $$
\boldsymbol{X}^\top_{\hat{\boldsymbol{\beta}  },(j)}\mathbf{X}_{\hat{\boldsymbol{\beta}  }} \hat{\boldsymbol{\Theta}}_{\hat{\boldsymbol{\beta}  },j}/n = 1, \quad \text{and}  \quad \Vert \mathbf{X}^\top_{\hat{\boldsymbol{\beta}  },(-j)}\mathbf{X}_{\hat{\boldsymbol{\beta}  }} \hat{\boldsymbol{\Theta}}_{\hat{\boldsymbol{\beta}  },j}/n \Vert_{\infty} = \Vert \nabla Q_{\lambda_{j}}(\hat{\boldsymbol{\varphi}}_{\hat{\boldsymbol{\beta}  },j}) \Vert_{\infty} / \hat{\phi}^2_{\hat{\boldsymbol{\beta}  },j},
 $$  which further indicates that 
 $$
\Vert \hat{{\boldsymbol{\Sigma}}}_{\hat{\boldsymbol{\beta}  }} \hat{\boldsymbol{\Theta}}_{\hat{\boldsymbol{\beta}  },j} - \boldsymbol{e}_{j}\Vert_{\infty} = \Vert \nabla Q_{\lambda_{j}}(\hat{\boldsymbol{\varphi}}_{\hat{\boldsymbol{\beta}  },j}) \Vert_{\infty} / \hat{\phi}^2_{\hat{\boldsymbol{\beta}  },j},
 $$
holds uniformly in $j$, where $\boldsymbol{e}_j $ is the $j$-th column of an identity matrix. 
Equivalently, we have the following property, 
\begin{equation*}\label{lemma3.1}
  \Vert\hat{\boldsymbol{\Theta}}_{\hat{\boldsymbol{\beta}  }}\hat{{\boldsymbol{\Sigma}}}_{\hat{\boldsymbol{\beta}  }}-\mathbf{I}\Vert_{\infty} 
  = \max_{j} \Vert \nabla Q_{\lambda_{j}}(\hat{\boldsymbol{\varphi}}_{\hat{\boldsymbol{\beta}  },j}) \Vert_{\infty} / \hat{\phi}^2_{\hat{\boldsymbol{\beta}  },j}.
 \end{equation*}
 Thus, by the Holder's inequality, it yields 
 \begin{equation*}
  {  \Vert{\boldsymbol{ \Delta}}^{(2)} \Vert_{\infty} \leq \Vert\hat{\boldsymbol{\Theta}}_{\hat{\boldsymbol{\beta}  }}\hat{{\boldsymbol{\Sigma}}}_{\hat{\boldsymbol{\beta}  }}-\mathbf{I}\Vert_{\infty}  \Vert \hat{\boldsymbol{\beta}} - \boldsymbol{\beta}^*\Vert_1} = O_p(\lambda^{**})\cdot O_p(s\lambda)= O_p(s\lambda\lambda^{**}).
  \end{equation*}

  Finally, with the condition $ s\lambda\lambda^{**} = o(n^{-1/2})$, we get the desired results in the theorem.
  

\end{proof}



\subsection {Proof of Theorem \ref{thm3}.}
\begin{proof}[\textbf{\upshape Proof of Theorem \ref{thm3}}]
  
  Recall the decomposition (\ref{de-bias01}),
  \begin{equation*}
  \begin{aligned}
  \hat{\boldsymbol{\beta}}  _{de} - \boldsymbol{\beta} ^*  &=  \hat{\boldsymbol{\Theta}}_{\hat{\boldsymbol{\beta}}} \mathbf{X}^\top\mathbf{W}_{\boldsymbol{\beta}  ^*}^2\boldsymbol{\epsilon}/n - (\hat{\boldsymbol{\Theta}}_{\hat{\boldsymbol{\beta}}} \mathbf{X}^\top\mathbf{W}_{\boldsymbol{\beta}  ^*}^2\boldsymbol{\epsilon}/n- \hat{\boldsymbol{\Theta}}_{\hat{\boldsymbol{\beta}}}\mathbf{X}^\top\mathbf{W}_{\hat{\boldsymbol{\beta}}  }^2\boldsymbol{\epsilon}/n) - (\hat{\boldsymbol{\Theta}}_{\hat{\boldsymbol{\beta}}}\hat{{\boldsymbol{\Sigma}}}_{\hat{\boldsymbol{\beta}}  }-\mathbf{I})(\hat{\boldsymbol{\beta}}   - \boldsymbol{\beta} ^*)\\
  & :=\hat{\boldsymbol{\Theta}}_{\hat{\boldsymbol{\beta}}} \mathbf{X}^\top\mathbf{W}_{\boldsymbol{\beta}  ^*}^2\boldsymbol{\epsilon}/n - {\boldsymbol{ \Delta}}^{(1)} -  {\boldsymbol{ \Delta}}^{(2)}.
  \end{aligned}
  \end{equation*}
  

Since Theorem \ref{thm_add} shows that the terms $\sqrt{n}{\boldsymbol{ \Delta}}^{(1)}$ and $\sqrt{n}{\boldsymbol{ \Delta}}^{(2)}$ are negligible, we only need to verify the following two items to conclude the desired results:
  
  (i)  The asymptotic normality of ${\mathbf R}\hat{\boldsymbol{\Theta}}_{\hat{\boldsymbol{\beta}}} \mathbf{X}^\top\mathbf{W}_{\boldsymbol{\beta}  ^*}^2\boldsymbol{\epsilon}/\sqrt{n} $.
  
  (ii) The estimator of the asymptotic variance-covariance matrix $\hat{\boldsymbol{\Omega}}_{\mathbf R}$ is consistent with ${\boldsymbol{\Omega}}_{\mathbf R}$.

  
  To prove (i), note that 
 \begin{eqnarray*}
 \left\Vert {\mathbf R}\hat{\boldsymbol{\Theta}}_{\hat{\boldsymbol{\beta}}} \sum_{i=1}^n\boldsymbol{X}_i w^2_{{\boldsymbol{\beta} ^*},i}\epsilon_i/n
   - {\mathbf R}\boldsymbol{\Theta}_{\boldsymbol{\beta} ^*} \sum_{i=1}^n \boldsymbol{X}_i w^2_{{\boldsymbol{\beta} ^*},i}\epsilon_i/n\right\Vert_{\infty}  &\leq& \Vert \sum_{i=1}^n\boldsymbol{X}_i w^2_{{\boldsymbol{\beta} ^*},i}\epsilon_i/n \Vert_{\infty}
   \Vert{\mathbf R}\Vert_{ l_{\infty}} \Vert \hat{\boldsymbol{\Theta}}_{\hat{\boldsymbol{\beta}}} -\boldsymbol{\Theta}_{\boldsymbol{\beta} ^*}  \Vert_{l_{\infty}} \\
   &=& O_p( s^{**}\lambda^{**}\sqrt{\ln p / n })
 \end{eqnarray*}
 and that $ {\mathbf R}{\boldsymbol{\Theta}}_{\boldsymbol{\beta} ^*}\boldsymbol{X}_{i} w^2_{{\boldsymbol{\beta} ^*},i}\epsilon_i $, $i\in\{1,\ldots, n\}$ are i.i.d. sequences of $p_0$-dimensional vectors with zero mean and finite variance.
 Thus, given  $s^{**}\lambda^{**}\sqrt{\ln p/n} = o(n^{-1/2})$, 
 by the Lindeberg-Levy central limit theorem along with the Slutsky's lemma, it follows that
 \begin{equation*}\label{pf_th3_3}
  {\mathbf R}\hat{\boldsymbol{\Theta}}_{\hat{\boldsymbol{\beta}}} \mathbf{X}^\top\mathbf{W}_{\boldsymbol{\beta}  ^*}^2\boldsymbol{\epsilon}/\sqrt{n}\xrightarrow{d} \mathcal{N}_{p_0}\left(\bf0, {\boldsymbol{\Omega}}_{\mathbf R}\right) \quad \textit{with} \quad  {\boldsymbol{\Omega}}_{\mathbf R} = {\mathbf R}{\boldsymbol{\Theta}}_{\boldsymbol{\beta}^*} {\sf E} \left[ \boldsymbol{X}_i \boldsymbol{X}^\top_iw^4_{{\boldsymbol{\beta} ^*},i}\epsilon^2_i \right]{\boldsymbol{\Theta}}^\top_{\boldsymbol{\beta}^*}{\mathbf R}^\top.
\end{equation*}

For (ii),  note that we further assume  $ {\sf E}[ x^8_{ij}]  = O(1)$ and $K^4 \sqrt{\ln p / n } = o(1)$, then 
following a similar argument as in the proof of Lemma \ref{lemmaA.2}, we can show that    
\begin{equation*}
\left\Vert\hat {\mathbf Q}- {\mathbf Q}\right\Vert_{\infty}=\left\Vert\frac{1}{n}\sum^n_{i=1} \boldsymbol{X}_i \boldsymbol{X}^\top_iw^4_{{\boldsymbol{\beta} ^*},i}{\epsilon}_i^2  - {\sf E} \left[ \boldsymbol{X}_i \boldsymbol{X}^\top_iw^4_{{\boldsymbol{\beta} ^*},i}\epsilon^2_i \right]\right\Vert_{\infty}=O_p(\sqrt{\ln p / n } ),
\end{equation*}
where $ {\mathbf Q} = {\sf E} \left[ \boldsymbol{X}_i \boldsymbol{X}^\top_iw^4_{{\boldsymbol{\beta} ^*},i}\epsilon^2_i \right]$ and $\hat {\mathbf Q} = \frac{1}{n}\sum^n_{i=1} \boldsymbol{X}_i \boldsymbol{X}^\top_iw^4_{{\boldsymbol{\beta} ^*},i}{\epsilon}_i^2 $. Then we have
\begin{eqnarray*}
  && \Vert   {\mathbf R}\hat{\boldsymbol{\Theta}}_{\hat{\boldsymbol{\beta}}} \hat{\mathbf Q} \hat{\boldsymbol{\Theta}}_{\hat{\boldsymbol{\beta}}}^\top{\mathbf R}^\top
   - {\mathbf R} {\boldsymbol{\Theta}}_{\boldsymbol{\beta}^*} {\mathbf Q} {\boldsymbol{\Theta}}^\top_{\boldsymbol{\beta}^*}{\mathbf R}^\top \Vert_\infty\\
   &\leq &
      \Vert   {\mathbf R}\hat{\boldsymbol{\Theta}}_{\hat{\boldsymbol{\beta}}} \hat{\mathbf Q} \hat{\boldsymbol{\Theta}}_{\hat{\boldsymbol{\beta}}}^\top{\mathbf R}^\top
   - {\mathbf R} {\boldsymbol{\Theta}}_{\boldsymbol{\beta}^*} \hat{\mathbf Q} {\boldsymbol{\Theta}}^\top_{\boldsymbol{\beta}^*}{\mathbf R}^\top \Vert_\infty
   +  \Vert   {\mathbf R}{\boldsymbol{\Theta}}_{\boldsymbol{\beta}^*}  (\hat{\mathbf Q}-{\mathbf Q}) {\boldsymbol{\Theta}}^\top_{\boldsymbol{\beta}^*}{\mathbf R}^\top
   \Vert_\infty\\
   &\leq& {   \Vert\hat{\mathbf Q}\Vert_\infty \Vert {\mathbf R}\hat{\boldsymbol{\Theta}}_{\hat{\boldsymbol{\beta}}} - {\mathbf R} {\boldsymbol{\Theta}}_{\boldsymbol{\beta}^*} \Vert_{l_{\infty}}^2 +2 \Vert\hat{\mathbf Q}\Vert_\infty \Vert {\mathbf R}\hat{\boldsymbol{\Theta}}_{\hat{\boldsymbol{\beta}}} - {\mathbf R} {\boldsymbol{\Theta}}_{\boldsymbol{\beta}^*} \Vert_{l_{\infty}}\Vert  {\mathbf R} {\boldsymbol{\Theta}}_{\boldsymbol{\beta}^*} \Vert_{l_{\infty}}
   + \left\Vert\hat {\mathbf Q}- {\mathbf Q}\right\Vert_{\infty}\Vert  {\mathbf R} {\boldsymbol{\Theta}}_{\boldsymbol{\beta}^*} \Vert_{l_{\infty}}^2}\\
   &=& O_p( (s^{**})^2(\lambda^{**})^2) + O_p( s^{**}\lambda^{**}){ O_p(\sqrt{s^{**}})} + O_p( { {s^{**}}}\sqrt{\ln p / n })\\
   &=& O_p({(s^{**})}^{3/2}\lambda^{**}),
\end{eqnarray*}
which complete the proof. 
\end{proof}

\bibliographystyle{elsarticle-harv} 
\bibliography{expectile.bib}





\end{document}